\documentclass[journal]{IEEEtran}%
\IEEEoverridecommandlockouts
\usepackage{soul}
\usepackage{cite}
\usepackage{amsmath,amssymb,amsfonts}
\usepackage{amsthm,bm}
\usepackage{colortbl}
\usepackage{hhline}
\usepackage{graphicx}
\usepackage{subfigure} 
\usepackage{textcomp}
\usepackage{url}
\usepackage{stfloats}
\usepackage{mathtools}
\usepackage{cuted}
\usepackage{array}  %表格上下左右居中
\usepackage {threeparttable}
\usepackage{tablefootnote}
\usepackage{multirow}
\usepackage{colortbl}
\def\BibTeX{{\rm B\kern-.05em{\sc i\kern-.025em b}\kern-.08em
    T\kern-.1667em\lower.7ex\hbox{E}\kern-.125emX}}

\usepackage{stackengine}
\usepackage{enumerate}
\usepackage{booktabs}
\usepackage{subfig}
\usepackage{multicol}
\usepackage[noend]{algpseudocode}
\usepackage[table,xcdraw]{xcolor}
\usepackage{pgfplots}
\usepackage{pgfplotstable}
\usepackage[ruled,linesnumbered]{algorithm2e} 
\SetKwRepeat{Do}{do}{while}
\definecolor{Reds}{RGB}{0,0,0}%{0,47,200}%{26,42,255}%{5,39,175}40,69,122
\definecolor{Blues}{RGB}{0,0,0}
\usepackage{setspace}
\usepackage{color}
\usepackage{tabularray}
\definecolor{Iron}{rgb}{0.811,0.815,0.815}
\makeatletter
\def\BState{\State\hskip-\ALG@thistlm}
\makeatother

\algnewcommand\algorithmicforeach{\textbf{for each}}
\algdef{S}[FOR]{ForEach}[1]{\algorithmicforeach\ #1\ \algorithmicdo}
\begin{document}

\title{Task-driven Semantic-aware Green Cooperative Transmission Strategy for Vehicular Networks
}
\author{Wanting Yang, Xuefen Chi, Linlin Zhao, Zehui Xiong, Wenchao Jiang \thanks{
This research is supported by National Natural Science Foundation of China under Grant 62271228,
the Jilin Scientific and Technological Development Program under Grant 20230101063JC and the China Scholarship Council CSC NO. 202106170088.
This research is also supported by the National Research Foundation, Singapore, and Infocomm Media Development Authority under its Future Communications Research \& Development Programme. The research is also supported by the SUTD SRG-ISTD-2021-165, the SUTD-ZJU IDEA Grant (SUTD-ZJU (VP) 202102), and the Ministry of Education, Singapore, under its SUTD Kickstarter Initiative (SKI 20210204).

W.~Yang is with the Department of Communications Engineering, Jilin University, Changchun, China, and also with Information Systems Technology and Design Pillar, Singapore University of Technology and Design, Singapore. Email: yangwt18@mails.jlu.edu.cn. 
X. Chi, and L. Zhao are with the Department of Communications Engineering, Jilin University, Changchun, China. Email: chixf@jlu.edu.cn, zhaoll13@mails.jlu.edu.cn
Z.~Xiong and W. Jiang are with Information Systems Technology and Design Pillar, Singapore University of Technology and Design Singapore. Email: zehui\_xiong@sutd.edu.sg. wenchao\_jiang@sutd.edu.sg 
}
}

\makeatletter
\setlength{\@fptop}{0pt}
\makeatother

\maketitle

\vspace{-1.8cm}
\begin{abstract}
Considering the infrastructure deployment cost and energy consumption, it is unrealistic to provide seamless coverage of the vehicular network. The presence of uncovered areas tends to hinder the prevalence of the in-vehicle services with large data volume. To this end, we propose a predictive cooperative multi-relay transmission strategy (PreCMTS) for the intermittently connected vehicular networks, fulfilling the 6G vision of  semantic and green communications. Specifically, we introduce a task-driven knowledge graph (KG)-assisted semantic communication system, and model the KG into a weighted directed graph from the viewpoint of transmission. Meanwhile, we identify three predictable parameters about the individual vehicles to perform the following anticipatory analysis. Firstly, to facilitate semantic extraction, we derive the closed-form expression of the achievable throughput within the delay requirement. Then, for the extracted semantic representation, we formulate the mutually coupled problems of  semantic unit assignment and predictive relay selection as a combinatorial optimization problem, to jointly optimize the energy efficiency and semantic transmission reliability. {\color{Blues}To find a favorable solution within limited time,} we proposed a low-complexity algorithm based on Markov approximation. The promising performance gains of the PreCMTS are demonstrated by the simulations with realistic vehicle traces generated by the SUMO traffic simulator.
% , which is solved by a multi-threaded search algorithm. The high energy efficiency and reliability of PreCMTS have been demonstrated with simulations.
\end{abstract}

\begin{IEEEkeywords}
Vehicular network, store--carry--forward, proactive cooperative transmission, semantic--aware, Markov approximation
\end{IEEEkeywords}

\newtheorem{definition}{Definition}
\newtheorem{lemma}{Proposition}
\newtheorem{theorem}{Theorem}

\newtheorem{property}{Property}

\vspace{0cm}

\section{Introduction}

\vspace{0cm}

\IEEEPARstart{T}{he} burgeon of the intelligent transportation system has spawned numerous innovative in-vehicle services to make mobility much safer and easier. Most services such as road sign recognition and situation understanding  heavily rely on scene understanding~\cite{luettin2022survey,kim2021accelerating,halilaj2021knowledge}, which are characterized by the large data volume. However, due to the high cost of infrastructure deployment and energy consumption, it is unrealistic to install sufficient
roadside units (RSUs) to provide seamless coverage~\cite{liu2022elastic}. The presence of low-throughput intermittently connected vehicular networks (ICVNs) inevitably hinders the~popularity of~these~services.

Thanks to the boom in artificial intelligence, semantic communication (SemCom)\footnote{{\color{Blues}In} our work, SemCom refers to the communication that reaches the semantic level or the effectiveness level, that are defined by Weaver in~\cite{weaver1953recent}. } has evolved from a theoretical concept to a 6G enabler. Exploiting the intelligence of vehicles, SemCom can achieve a significant reduction in transmission burden, thus mitigating the impact of ICVNs on quality of service. For example, deep learning (DL) is now commonly used to perform human-like understanding at transmitters and  receivers, which are termed as semantic encoding and semantic decoding, respectively~\cite{yang2022semantic}. Therein, the irrelevant information about the target communication task is filtered out before transmission, and only a small data volume carrying valuable information is transmitted to receivers for the downstream inference task~\cite{9979702}. The promising performance gains achieved by SemCom in low channel conditions has been widely demonstrated~\cite{lee2019deep,xie2021deep,weng2021semantic}. 
Nonetheless, the black box nature of the DL-based SemCom results in low social acceptance. Moreover, the focus of existing SemCom research is mostly on the semantic processing of transceivers, where  wireless environment  is simplified to a channel model,  such as Rayleigh channel and  Rician channel.
This makes it infeasible for dynamic complex vehicle networks, where the average channel gain experienced by users changes significantly.  Thus, an explainable and generalized SemCom paradigm is called~for.

Fortunately, the recent studies on  the convergence of knowledge graph (KG) and explainable computer vision, holds promise towards the mentioned expectations. KG  can  provide a structured semantic representation (SR) for road traffic scenes~\cite{luettin2022survey}, which can be seen as a container of semantic information. In contrast to the underlying raw data formats of the practical scene,  the great extensibility of KG allows the KG-based SR to be partially updated according to the dynamic changes of  scenes, e.g., new roadblocks~\cite{qiu2020knowledge}. Meanwhile, the semantic information for different target tasks can be flexibly extracted in form of a sub-KG. For instance, for users interested in the road traffic, 
only the sub-KG related to pedestrian and traffic flow needs to be transmitted, and building-related sub-KG can be automatically filtered out.  
Nonetheless, a completed end-to-end model of a universal KG-based SemCom system is still a gap in the available research.
% More promisingly, the techniques about scene graph generation, KG embedding, and KG reasoning have been widely investigated for diverse tasks~\cite{zhu2022scene,luettin2022survey}, which have laid the foundation for the KG-based semantic encoding and decoding. 
Furthermore, the distinctive feature of SemCom  lies in that data is assigned diverse significance~\cite{yang2022semantic}. 
From a well-established KG,  both the semantic importance of each semantic unit (SU) and the number of the bits required to carry SU viewed from the physical form can be obtained,  which  creates the opportunity for the finer-grained semantic-aware transmission strategy design.
 For instance, the SUs of greater importance can be transmitted with higher power, wider bandwidth, or more reliable links~\cite{liew2022economics} to enhance the semantic transmission reliability. However, this cannot be easily realized by straightforward refinements to existing schemes.

% What's more, viewed from the technical level, the number of bits required to carry each semantic unit (SU) can be obtained form the KG, which also creates the opportunity for the strategic optimization of semantic-aware resource allocation.
% More promisingly, 
% the techniques about  both the scene graph generation~\cite{halilaj2021knowledge,luettin2022survey,zhu2022scene},  and KG reasoning  for diverse tasks  such as scene regeneration, image captioning, and question answering~\cite{kim2021accelerating,zhang2022knowledge,luettin2022survey}  have been widely investigated, which have laid the foundation for the KG-based semantic encoding and decoding. 
% Nonetheless, a completed model of a universal KG-based SemCom system is still a gap in the available research on SemCom.

% Furthermore, the distinctive feature of SemCom  lies in that data is assigned diverse significance~\cite{yang2022semantic}. As such, a semantic-aware transmission strategy with a finer granularity of resource management than the flow level can further enhance the semantic reliability of transmission. For instance, the semantic units of greater importance can be transmitted with higher power~\cite{liew2022economics},  or more reliable links. However, this cannot be easily realized by straightforward refinements to existing schemes. 

Especially, for ICVNs, most research efforts are devoted into the multi-hop relay transmission for lightweight services with strict delay requirements~\cite{liu2022mobility,zhang2021efficient},~where~the relays serve for real-time amplification/decoding and forwarding.
If they are applied to the services with large volume, much unwarranted communication overhead and cache pressure on the relays are introduced. Given this, the \textit{one-hop} store–carry–forward scheme (SCFS)~\cite{kolios2010load} is more appropriate for the considered case, where  the mobility of relay vehicles is to utilized to physically propagate information messages to reduce the outage areas~\cite{kolios2010load}. 
However, the existing studies on SCFS only concentrate on the physical layer, such as minimizing the outage time~\cite{wang2016cooperative}, statistical analysis of achievable throughput gain~\cite{liu2022elastic}.
Few of them care about the properties of the communication task, even for the content size and maximum acceptable delay. As a result, they cannot achieve on-demand~fulfillment and are further away from semantic awareness. More critically, as the energy efficiency varies greatly depending on the vehicle location, the total energy consumption is strongly related to the selected relays. Thus, these ready-made SCFSs with uniform relay selection rules tend to cause different levels of energy waste for a specific task, depending on real-time on-road vehicles' location and speed. This goes against green communication in 6G.

In light of the above, to meet the 6G vision of SemCom and green communication, we propose a novel predictive cooperative multi-relay transmission strategy (PreCMTS) for large  download for ICVNs. In the  strategy, the SR selection, relay selection, and SU assignment are all highly related to three predictable parameters: the residual dwell time of the vehicles in their associated RSUs, as well as the encounter time and V2V link lifetime of each relay with the target vehicle. The major contributions are highlighted as follows.

\begin{itemize}
    \item We introduce a general task-driven KG-assisted SemCom system model, where both the semantic encoding and decoding are performed based on the KG. Moreover, from the standpoint of transmission, we
model the KG as a weighted directed graph (wDG), where the vertices represent the indivisible SUs that are the embedding of real-world objects and their abstract relationships, and the directed edges characterize the dependence of SUs. To enable semantic-aware transmission, the  significance degree and the data size viewed from semantic and physical level for each SU are recorded as edge weights.  
    \item To facilitate semantic extraction (SE) to get an appropriate SR, we derive the closed-form expression of the achievable throughput within the maximum acceptable delay according the current road traffic situation. Moreover, for the selected SR, we formulate the  mutually coupled problems of predictive relay selection and SU assignment as a combinatorial optimization problems with the aim to minimize energy consumption while guaranteeing the  semantic transmission reliability under imperfect speed prediction. Therein, the constraints of the V2V link interference, the end-to-end transmission  delay\footnote{The end-to-end delay in our work refers to the time interval between the moment when the target vehicle sends request and the moment when the target vehicle receives all the requested data.}, and the bottleneck of the two cascade store-carry-forward links are all considered.
    % \item  {\color{Blues}For a given SR with the data volume less than the achievable throughput,} we jointly formulate the relay selection and SU assignment as a binary integer programming problem with the aim to minimize energy consumption and maximize semantic reliability. {\color{Blues}Meanwhile, the mathematical formulation captures the conflicts on different V2V links, the integrity of transmission per link, and the total transmission delay of any possible strategy.}
    %Meanwhile, the constraints of  the interference among the V2V links,  and the bottleneck of the two-hop store-carry-forward link are all considered.
    \item {\color{Blues}To find a favorable solution within limited time,} we design a low-complexity multi-threaded search algorithm based on Markov approximation. Moreover, we devise an SU assignment algorithm following the basic SCFS in~\cite{liu2022elastic} as a baseline to generate the initial state. From the simulation results,  in PreCMTS, the SUs with high semantic significance are more likely assigned to the direct transmission link to enhance semantic reliability and the vehicles close to the RSU are preferred to be selected as relays to pre-stores SUs to save energy compared to the baseline. The promising performance gains in terms of energy saving, semantic transmission reliability, and semantic energy efficiency are demonstrated.
\end{itemize}

In the following sections, we first review the related works about SemCom and SCFS, respectively. Then, we describe the system model and highlight the overview of the proposed semantic-aware PreCMTS in Section III. Then, the  details of the proposed scheme  are presented in Section IV. Section V presents the simulation results, and Section VI concludes this paper. Besides, the notations of relevant parameter symbols are listed in Table I.

\begin{table*}
\footnotesize
 \centering
 \caption{{ List of relevant notations.}}
 \begin{tabular}{|c|m{6cm}||c|m{5.5cm}|}
  \hline
 Notation & Description & Notation & Description  \\
  \hline
  \hline
${r_{\text{I}}}\left( {{r_{\text{V}}}} \right)$ & Communication radius of RSU (vehicle) &  ${R_{{\text{I}}}}\left( {{R_{{\text{V}}}}} \right)$ & Data transmission rate of V2I (V2V) link \\ \hline
$D_i^{\text{I}}$ & Maximum duration of V2I link for vehicle ${v_i}$  & ${T_{\max }}$ & Maximum acceptable delay \\ \hline
$\Delta _i^{\text{T}}$ & Moment when relay vehicle $v_i^{\text{R}}$  enters the communication range of the target vehicle & $D_{i}^{{\text{T}}}$ & Maximum duration of the V2V link between target vehicle and relay vehicle $v_i$\\ \hline
$\hat C_i^{{\text{I}}}$ & Maximum data amount that can be transmitted via V2I link to vehicle $v_i \in {\cal V}$ & ${\hat \delta _i}$ & Moment when relay $v_i$ starts to forward data in  achievable throughput analysis \\ \hline
$ \hat C_i^{{\text{V}}}$ & Maximum data amount transmitted to vehicle $v_0$ by relay $v_i \in {\cal V}_{\text{R}}$ for a given ${\mathbf{\Phi }}$ & $t_i^{{\text S}_{\text v}}$ & Moment when relay $v_i$ starts to forward data to $v_0$ for a given transmission strategy\\ \hline
${\beta _j}$ & Data size of SU $j$   & ${\alpha _j}$  & Contribution of SU $j$ to the accuracy of SR\\ \hline
 \end{tabular}
 \label{tbl:mae}

\end{table*}

\section{Related works}

\vspace{0cm}
\subsection{Semantic Communication}
Based on our previous review works~\cite{yang2022semantic,9979702}, the existing research on SemCom can be broadly classified into four categories depending on the SE method. The first and most studied category is deep-learning (DL)-based end-to-end SemCom.
The employed semantic encoder and decoder are two separate learnable neural networks, and linked through a layer for modeling random channels~\cite{xie2020lite}. They are trained jointly, based on a complete data set shared by both senders and receivers. Thanks to the advancement of  DL models, e.g., Transformer, the high efficient SE for  text, image, and audio, achieves significantly performance gains especially at  low signal-to-noise ratio (SNR) region~\cite{lee2019deep,zhou2021semantic,weng2021semantic}. 
Nevertheless, the back-propagation in DL paradigm requires the loss function to be differentiable, which hinders the sophisticated non-differentiable semantic metrics from being applied to guide the training. To solve this issue,  deep reinforcement learning paradigm is adopted to perform SE~\cite{lu2022rethinking}. However, the above two categories of SemCom are  available only for the simplest point-to-point communication model, which cannot be directly applied to complex real-world scenarios. Moreover, the black box nature also restricts their social acceptance~\cite{yang2022semantic}. 

Meanwhile, with the development of the explainability of AI technologies, some researchers propose the knowledge base (KB)-assisted SemCom. Herein, the KB is a special database for semantic knowledge management, which consists of semantic elements embedded in the source data, the involving communication tasks, and the possible ways of reasoning  by communication participants~\cite{9979702}. 
 Up to now, there are two available kinds of general KB models. One is based on a hierarchical structure~\cite{farshbafan2022curriculum}, and the other is based on graph structure~\cite{thomas2022neuro,9979702}. 
 % Nonetheless, the existing literature tends to depict the KB model from the perspective of semantic representations, without considering the transmission process. 
 Moreover,  there have been several technical research~\cite{wang2022performance,zhou2022cognitive} on SemCom for text transmission based on the available the interconversion technologies for text and graph. However, in the above works, the resource allocation algorithm is still following the philosophy of traditional \textit{content-blind} resource allocation paradigm, i.e., allocating radio resources to per user according to their required data volume. 
In this sense, a semantic-aware transmission has not been achieved in a real sense.

In addition to the above three categories of SemCom, there is also a semantic-native SemCom paradigm, wherein the semantic information can be learned from iterative communications between intelligent agents~\cite{seo2021semantics}.
However, this study is still stuck in the theoretical analysis based on a simple ideal model. It remains a huge challenge to put it into practice. In this sense, we focus on the KB-assisted SemCom in our work.

\vspace{0cm}

\subsection{Store-Carry-Forward Scheme}
The core concept of SCFS is to utilize the mobility of relays  to physically propagate information messages~\cite{kolios2010load}, which is first proposed in~\cite{fall2003delay}. 
Specifically, in SCFS, a relay vehicle, which will pass the target vehicle within their uncovered area, pre-stores partial data requested by the target vehicle in advance over an available V2I link. It then carries and forwards data to the target vehicle until they encounter each other. 

Initially,
in~\cite{wu2013adaptive,bouk2015outage}, the authors investigate the optimization of the target vehicle speed control with the objective of minimization of the outage time.  However, it deviates from the design philosophy of user experience-oriented communication nowadays, and it is unrealistic to control the vehicle's speed without consideration of the actual traffic conditions and the driver's driving habits.   Additionally, the above works only focus on the unidirectional road model, and thus the mobility pattern of the vehicles is not fully exploited. 
%For example, it is much easier to forecast the encounter of two vehicles traveling in opposite directions than those in the same direction.
To that end, a  bidirectional road is considered in~\cite{trullols2011cooperative,wang2016cooperative}, where the mobility
of vehicles can be utilized to physically propagate information messages.  Different from \cite{wu2013adaptive,bouk2015outage}, the authors in \cite{trullols2011cooperative,wang2016cooperative} propose some essential relay selection constraints on the relay link lifetime, residual dwell time, and buffer time, which can jointly determine which candidate vehicles can be picked as relays. All the above works just focus on the minimization of the outage time.

In~\cite {chen2016achievable,liu2022elastic}, the authors derive a closed-form expression of the
achievable throughput of the SCFS for a bidirectional road with vehicle flow obeying Poisson distribution. 
In~\cite{chen2016achievable}, two assumptions are made. The first one is that there is no possibility that the relay vehicle is still within the available V2V connection range, but it has no data to forward to the target vehicle. The second one is that there is no interference between different V2V links, i.e., the target vehicle can maintain multiple V2V links simultaneously. Furthermore,  in~\cite{liu2022elastic}, the authors propose an elastic-segment-based V2V/V2I cooperative strategy, where the second assumption is removed, and  a commonly used interference model~\cite{agarwal2004capacity} is adopted, that is, only one V2V link is active at any given time. The adopted assumption is strongly dependent on the
specific scenario, and thus compromising the generality of
their work.

Moreover, the existing works focus on the enhancement and evaluation of the physical layer performance. Few of them considers the demand of communication task and the energy efficiency of the communication system. To this end, a task-oriented SCFS is promising to embrace the green communication in 6G with on-demand fulfillment.

% The throughput gain of SCFC  considering different interference models. In~\cite{chen2016achievable}, the authors assume that there is no interference between different V2V links and in~\cite{liu2022elastic}, the authors assume that only one V2V link is active at any given time.

\vspace{0cm}

\section{{System  overview}}
\subsection{{Scenario Description}}
\label{Description}
\begin{figure*}[t]
 \centering
\includegraphics[scale = 0.68]{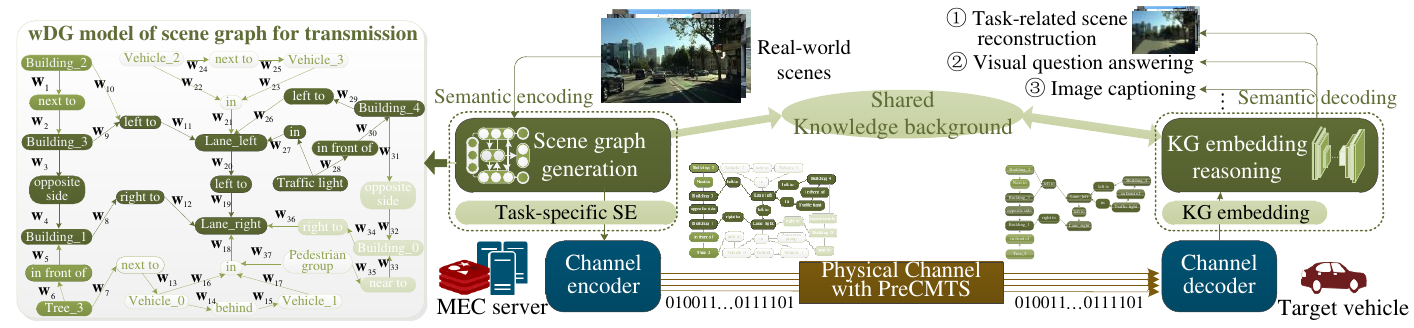}\\
 \caption{KG-assisted SemCom system model. 
 }
\label{KGmmodel}

\end{figure*}

This paper focuses on one segment of a bidirectional road, which runs through the coverage of two adjacent RSUs (indexed by RSU A and RSU B, respectively). 
The distance between the two RSUs is denoted by $H$, and the coverage radius of each RSU is denoted by $r_{\text I}$. Considering the restricted transmit power and the high deployment costs, we assume that there is an outage area between the  RSUs, i.e., $H > 2{r_{\text{I}}}$. 

 Without loss of generality, we assume that a vehicle within the coverage of RSU~A  sends a  request to multi-access edge computing (MEC) server for a large download. The maximum allowable delay of the services is denoted by ${T_{\max }}$. To complete the transmission within ${T_{\max }}$, we propose a  semantic-aware PreCMTS, which is performed by a central controller (CC) at the MEC server. The vehicle sending the request is referred to as the target vehicle and denoted by ${v_{\text{0}}}$. The relay candidates are the vehicles driving in the opposite direction to the target vehicle within the coverage of RSU~B.  The set of the relay candidates is  denoted by ${\mathcal{V}_\text{R}}$, and each of them is indexed by $v_i \in {\mathcal{V}_\text{R}}, i \in \left\{ {1,2, \ldots ,\left| {{\mathcal{V}_{\text{R}}}} \right|} \right\}$. Specifically, the serial number of the relay candidate is arranged according to the sequence of them entering the target vehicle's communication range.
 Similar to the RSUs, the communication range is the same for all vehicles, the radius of which is denoted by $r_\text{V}$. Since the RSU typically has stronger communication capability than the vehicle, we have $r_\text{I} > r_\text{V}$~\cite{chen2016achievable}.  Moreover, for ease of reference, we denote the set composed by the target vehicle and the relay candidates by $\mathcal{V}$, i.e., $\mathcal{V} = {\mathcal{V}_{\text{R}}} + \left\{ {{v_0}} \right\}$. In addition, we assume that the average speed remains constant~\cite{liu2022elastic}, and the  average speed of each vehicle ${v_i} \in \mathcal{V}$ is denoted by ${{\bar u}_i},i \in \left\{ {0,1,2, \ldots ,\left| {{\mathcal{V}_{\text{R}}}} \right|} \right\}$. To facilitate an energy efficient scheme, the vehicles are required to report the information about their speed and  position  to the MEC server, which enables the possibility of predictive relay selection and strategic pre-store the data in the relays under better channel states.
% Furthermore, we assume that both the target vehicle and the MEC server enable task-driven KG-assisted SemCom shown in Fig.~\ref{KGmmodel}.

\subsection{{Transceiver Semantic Processing Model}}
\label{SKGmodel}
The proposed task-driven KG-assisted SemCom system model is shown in Fig.~\ref{KGmmodel}, where a  two-dimensional image is taken as an example of semantic encoding input. 

The semantic encoding performed at the MEC server consists of two modules. Firstly, the \textit{scene graph generation} module  bridges the gap between visual and semantic perception of the real-world scene\footnote{It typically undergoes four steps: {off-the-shelf object detectors}, {feature representation}, {feature refinement}, and {relationship prediction}~\cite{zhu2022scene}}.  Then, the well-developed scene graph, i.e., a KG, can be regarded as a container for all the semantic information implied by the scene.
It is composed of multiple linked triples in the form of $\left\langle {head\_object,relation,tail\_object} \right\rangle$, e.g., $\left\langle {building1,right,lane1} \right\rangle $. The embedding\footnote{The embedding is a low-dimensional vector (being in accord with Word2vec in NLP)~\cite{zhang2022knowledge}, which are obtained in the KG generation via an visual translation embedding methods~\cite{zhu2022scene}.} of each element in each triples are treated  as an undividable SU. 
To facilitate semantic-aware
transmission,  the scene graph is re-modeled as a mathematical form of wDG as shown on the left side of Fig.~\ref{KGmmodel}. The SUs are treated as the vertices. The directed edges retain the dependency between the two objects. 
In general, the significance of SUs varies for different tasks. For instance, for users who intend to check the map of a certain place, information about pedestrian and traffic flow on the road is no longer necessary; on the contrary, for users who prefer to know the road traffic, detailed information about the surrounding buildings can be ignored. {\color{Blues}Therefore, we assign an array, ${{\bf w}_j} = \left[ {{w_{j,1}}, \ldots, {w_{j,k}},\ldots,{w_{j,K}}} \right]$, as the weight corresponding to an SU $j$, where $K$ represents the number of the tasks and $w_k$ is in form of a binary tuple ${{w_{j,k}} = \left( {{\alpha _{j,k}},{\beta _{j,k}}} \right)}$, with ${\alpha _{j,k}}$ denoting the quantified importance degree of SU $j$ to  task $k$ and ${\beta _{j,k}}$ denoting the number of bits required to carry the information of SU $j$. 
Without loss of generality, only one task $k$ is considered in our work. 
 To simply the notation, the subscript $k$ is omitted in this manuscript, and the weight for SU $j$ is simplified to ${{w_{j}} = \left( {{\alpha _{j}},{\beta _{j}}} \right)}$.}
Then,  based on the wDG, the \textit{task-specific SE}  module extracts an SR in form of {\color{Blues}an edge-induced subgraph} of the original wDG. 
{\color{Blues}To ensure the completeness of the transmission within the allowable maximum delay $T_{\text {max}}$, the edges with higher semantic importance have priority to be added 
to the edge subset used to generate the SR, while ensuring that the total data size is less than the achievable throughput. On this premise, the cardinal number of edge subset, which determines the number of the chosen SUs, can be decided based on a specific trade-off between semantic accuracy and energy consumption.}
 Moreover,  we assume that each SU in the selected SR is encapsulated individually  according to the edge weight ${{\bf{w}}_j} = \left( {{\alpha _j},{\beta _j}} \right)$~\cite{zhang2022toward}. That is, during the transmission, the data for each SU cannot be further split.

\begin{figure*}[t]
 \centering
\includegraphics[scale = 0.25]{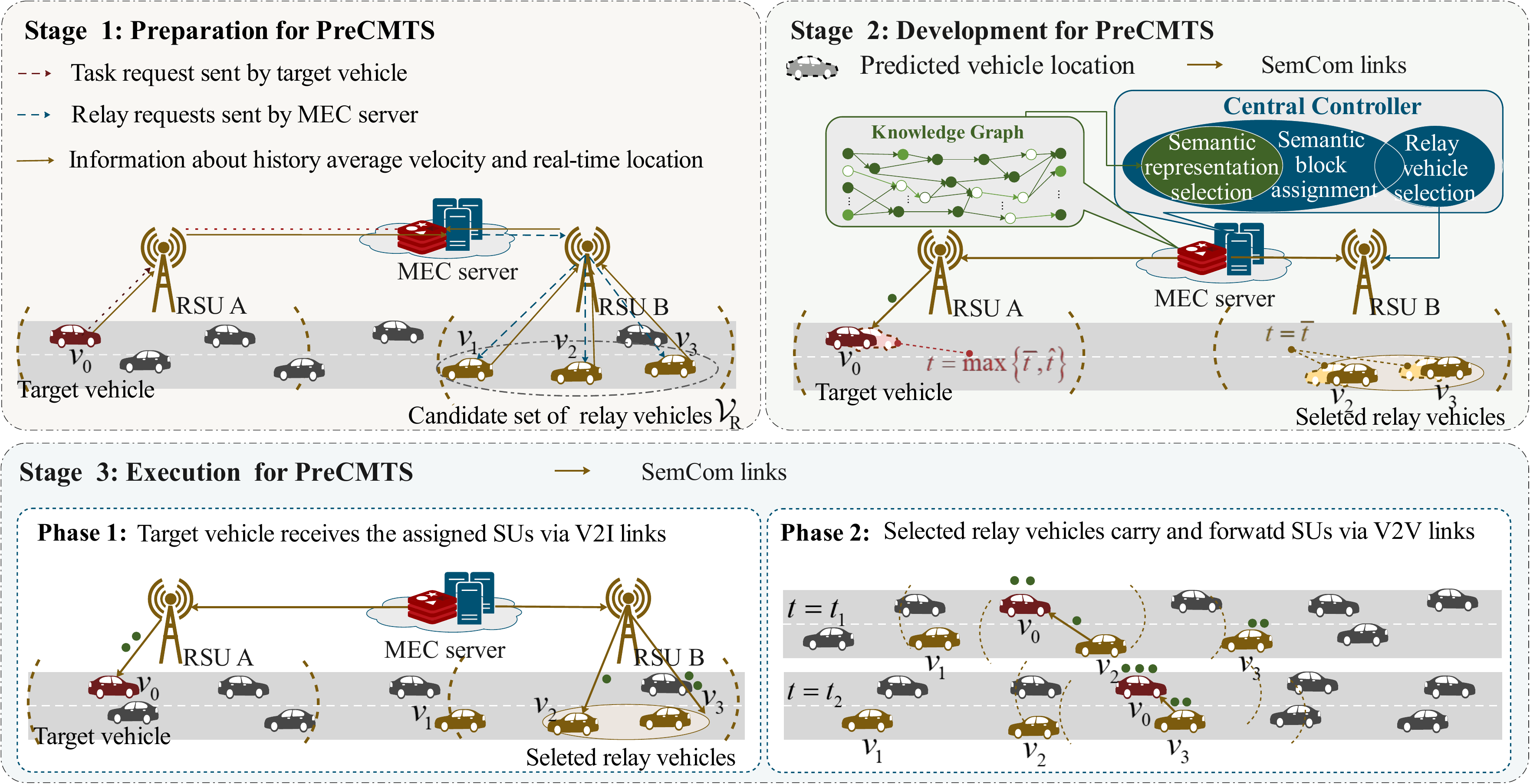}\\
 \caption{\color{Blues}Overview of the semantic-aware PreCMTS. 
 }
 \label{systemmodel}

\end{figure*}

Upon obtaining the complete task-related sub-KG, the target vehicle performs the semantic decoding, which is also accomplished by two modules. First, the \textit{KG embedding} module is responsible for embedding the objects and relations of the  sub-KG into a low-dimensional vector~\cite{zhu2022scene}. Then, taking the low-dimensional vector as the input, the \textit{KG-embedding reasoning} module performs the downstream semantic inference based on a cascaded sophisticated  network specially designed for the particular task, such as task-related scene reconstruction, visual question answering, and image captioning~\cite{zhu2022scene,zhang2022knowledge}.

It should be noted that, {\color{Blues}both sides of the communication are required to share their knowledge background about historical scenes and all the possible tasks, which allows the training process for semantic encoding and decoding to match each other.} Due to the limitation of space,  the  synchronization of  knowledge background are beyond the scope of this work.  Meanwhile, the communication overhead for the background knowledge, the computing resources for the KG generation and update, and the storage resources for the KG are not discussed in this paper but will be studied in the future works.

 %Meanwhile, a semantic representation of a task can be modelled as a chain of ordered semantic blocks, which is called a semantic chain for brevity and denoted by $\mathcal{J}$.
% In order to capture the dependencies between the individual semantic blocks and the correspondence between the semantic blocks and the individual tasks, a weighted Directed Acyclic Graph (wDAG) is employed to characterize the common KB, which is denoted by $\mathcal{G} = \left( {\mathcal{V},\mathcal{E}} \right)$. 

\vspace{-0.3cm}

\subsection{{\color{Blues}Wireless link transmission Model}}
\label{communicationmodel}
Without loss of generality, we assume that every passing vehicle is equipped with one antenna~\cite{liu2022elastic,chen2016achievable}, which allows a vehicle to maintain only one link at a time, either a V2I link or a V2V link. Meanwhile, to avoid interference, if a V2V link already exists within a vehicle communication range, it will not be able to transmit data~\cite{liu2022elastic}.
Moreover, considering the sophisticated technologies available in RSUs, such as frequency division multiplexing and multi-user beamforming, we assume that the RSU can simultaneously transfer data to multiple users without inter-user interference~\cite{liu2022elastic,wu2021v2v}. 

{\color{Blues}As depicted in Section~\ref{Description}, the extensively used disk model is employed  to characterize the V2I and V2V connection~\cite{liu2022elastic,chen2016achievable}. That is, any vehicle pair or vehicle-RSU pair is able to be connected if the distance between each other is less than $r_\text{V}$  or $r_\text{I}$~\cite{6180096}.  We denote the distance between any pair of transmitter and receiver by $d$.} The large-scale channel gain, then, can be characterized by the standard power-law path loss ${G_{x}}\left( {d} \right) = {b_{x}}{{d}^{ - {a_{x}}}}$, where  ${a_{x}}$ is the path loss exponent, ${b_{x}}$ is the reference path loss at a unit distance, and ${x} \in \left\{ {{\text{I,V}}} \right\}$ is set to differentiate the V2I and V2V links~\cite{su2022content,xu2021socially}. {\color{Blues}Furthermore, considering the high mobility of vehicles and the inevitable inter-vehicle large vehicle obstructions, e.g., buses, we adopt the $\mathcal{F}$ composite
fading model to characterize the small-scale fading, where the combined effects of multi-path and shadowing are taken into account~\cite{8638956}.  
We denote the small-scale channel gain by  $\tilde g$. Accordingly, the probability density function of  $\tilde g$ is expressed by~\cite{8638956}
\begin{equation}
    f\left( {\tilde g} \right) = \frac{{{m^m}{{\left( {{m_s} - 1} \right)}^{{m_s}}}{{\bar g}^{{m_s}}}{{\tilde g}^{m - 1}}}}{{B\left( {m,{m_s}} \right){{\left[ {m\tilde g + \left( {{m_s} - 1} \right)\bar g} \right]}^{m + {m_s}}}}},\label{pdf}
\end{equation}
where $m$, $m_s$, and ${\bar g}$ represents the number of clusters of multipath, shadowing shape, and average small-scale channel gain, respectively, and $B\left( { \cdot , \cdot } \right)$ denotes the beta function~\cite{8638956}. 
 We assume that the power control technique is adopted, where the decoding threshold of SNR for the V2I and V2V link are denoted by $\Upsilon _\text{I}$ and $\Upsilon _\text{V}$, respectively. Upon assuming perfect
capacity achieving coding, the achievable transmission rate is expressed by
\begin{equation}
    {R_x} = B_x\log \left( {1 + {\Upsilon _x}} \right), 
\end{equation}
where ${{x}} \in \left\{ {{\text{I,V}}} \right\}$. Moreover,
to simplify problem analysis, we assume that the bandwidth allocated to all the V2I links are fixed and the same~\cite{liu2022elastic}, i.e., ${B_{\rm{I}}} = {B_{\rm{V}}}$. 
 The instantaneous transmit power, then, is expressed as}
\textcolor{Blues}{
\begin{equation}
{p_x}\left( d \right) = \left\{ {\begin{array}{*{20}{c}}
{\frac{{{\Upsilon  _x}{\sigma ^2}}}{{{A_x}{G_x}\left( d \right)\tilde g}},}&{d \le {r_x}}\\
{0,}&{d > {r_x}}
\end{array}} \right.,
\label{P}
\end{equation}}
where ${{x}} \in \left\{ {{\text{I,V}}} \right\}$. Specifically, ${A_\text{I}}$ (${A_\text{V}}$) denotes the joint antenna gain of the transmitter and receiver of the V2I (V2V) link. Moreover, ${\Upsilon _{\text{I}}}$ (${\Upsilon _{\text{V}}}$) represents  {\color{Blues}the decoding threshold of signal-to-noise (SNR)} for the V2I (V2V) link.  {\color{Blues}Meanwhile, we assume that the small-scale gains are independently and identically
distributed (i.i.d.) among transmission time intervals. Then, the average transmit power with distance $d$ between the transmitter and receiver can be expressed by 
\begin{equation}
  \begin{aligned}
     &{{\bar p}_x}\left( d \right) = \mathbb{E}{_{\tilde g}}\left[ {{p_x}} \right] = \int_0^\infty  {\frac{{{\Upsilon _x}{\sigma ^2}}}{{{A_x}{G_x}\left( d \right)\tilde g}}f\left( {\tilde g} \right)d\tilde g}\\
     &=\frac{{{\Upsilon _x}{\sigma ^2}}}{{{A_x}{G_x}\left( d \right)}}\int_0^\infty  {{{\tilde g}^{ - 1}}f\left( {\tilde g} \right)d\tilde g}  \\
     & = \frac{{{\Upsilon _x}{\sigma ^2}}}{{{A_x}{G_x}\left( d \right)}}{\mathbb{E}}\left[ {{{\tilde g}^{ - 1}}} \right],\label{sm}
\end{aligned}  
\end{equation}
According to \eqref{pdf}, with the aid of \cite[eq. (3.194.3)]{Table}, the ${n^{{\rm{th}}}}$ moment of ${\tilde g}$ can be derived as 
\begin{equation}
{\mathbb{E}}\left[ {{{\tilde g}^n}} \right] = \frac{{{{\left( {{m_s} - 1} \right)}^n}{{\bar g}^n}\Gamma \left( {m + n} \right)\Gamma \left( {{m_s} - n} \right)}}{{{m^n}\Gamma \left( m \right)\Gamma \left( {{m_s}} \right)}}  \label{moment}
\end{equation}
 where  $\Gamma \left(  \cdot  \right)$ represents the gamma function. Substituting the case of $n=-1$ in \eqref{moment} into \eqref{sm}, we can obtain the final expression of  ${{\bar p}_x}\left( d \right)$ as below.
\begin{equation}
    {{\bar p}_x}\left( d \right) = \frac{{{\Upsilon _x}{\sigma ^2}}}{{{A_x}{G_x}\left( d \right)}}\frac{{m\Gamma \left( {m - 1} \right)\Gamma \left( {{m_s} + 1} \right)}}{{\left( {{m_s} - 1} \right)\bar g\Gamma \left( m \right)\Gamma \left( {{m_s}} \right)}}. \label{ap}
\end{equation}
For brevity, we rewrite \eqref{ap} as ${{\bar p}_x}\left( d \right) = {\rm{M}}\frac{{{\Upsilon _x}{\sigma ^2}}}{{{A_x}{G_x}\left( d \right)}}$, where ${\rm{M}} = \frac{{m\Gamma \left( {m - 1} \right)\Gamma \left( {{m_s} + 1} \right)}}{{\left( {{m_s} - 1} \right)\bar g\Gamma \left( m \right)\Gamma \left( {{m_s}} \right)}}$ is a constant.
 
 }

\begin{figure*}
\begin{minipage}[t]{0.48\linewidth}
    %\begin{figure}[t]
 \centering
 \includegraphics[scale = 0.54]{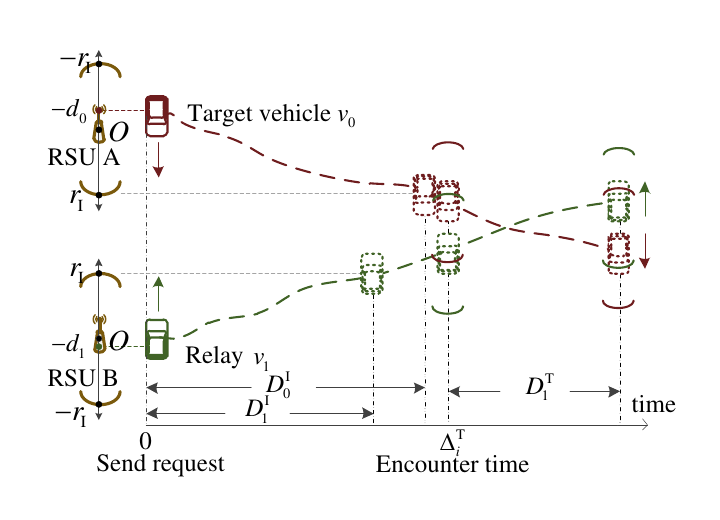}\\
 \caption{Diagram of vehicle encountering.
 }
 \label{meeting_diagram}
\quad
\end{minipage}
\begin{minipage}[t]{0.48\linewidth}
    %\begin{figure}[t]
 \centering
 \includegraphics[scale = 0.38]{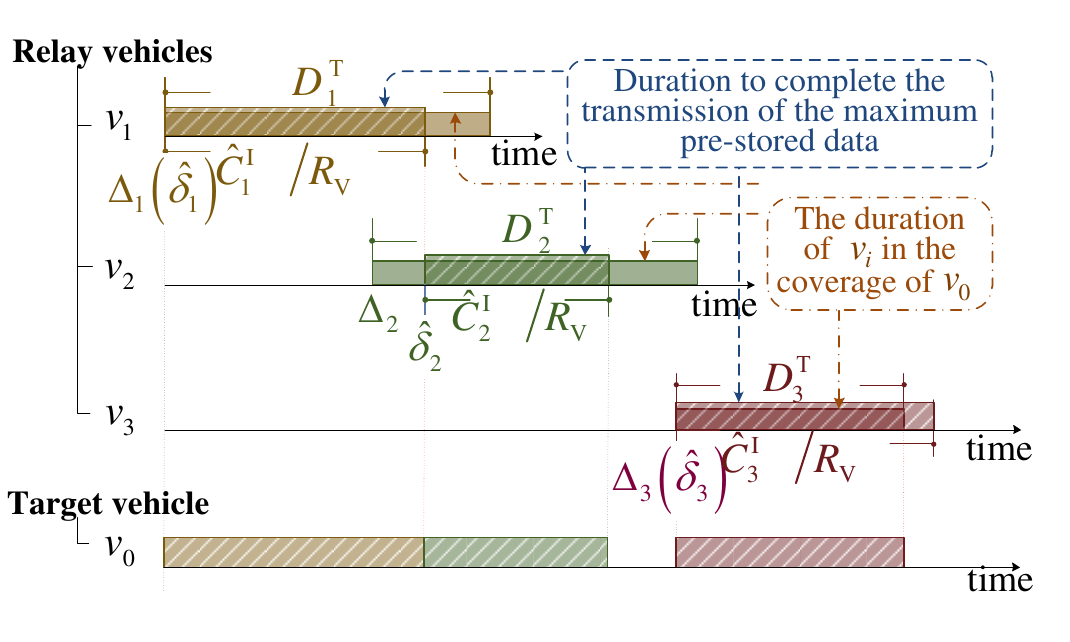}\\
 \caption{Illustration of V2V-link duration.
 }
 \label{max_V2V_duration}
%\end{figure}
\end{minipage}

\end{figure*}

% In general, the practical data transmission rate in the wireless channel under a certain SNR can be obtained by checking modulation and coding scheme lookup table~\cite{zhou2018resource}. However, for the conciseness of the presentation, we calculate the average achievable transmission rate of the V2I and V2V links based on Shannon's channel  capacity, that is,
% \begin{equation}
%     {R_x} = B_x\log \left( {1 + {\Gamma _x}} \right), 
% \end{equation}
% where ${{x}} \in \left\{ {{\text{I,V}}} \right\}$. 

% Despite the effects of Doppler shift and multipath, the average transmission rate of the V2I link $R_\text{V2I}$ (V2V link $R_\text{V2V}$) can be maintained at
% ${R_{{{y}}}} = B\log
%   \left( {1 + {\Gamma _{{{x}}}}} \right)$,  
% where $\left( {x,y} \right) \in \left\{ {\left( {{\text{I}},{\text{V2I}}} \right),\left( {{\text{V}},{\text{V2V}}} \right)} \right\}$ and $B$ represents the bandwidth allocated to the link.
% {\color{Blues}Considering that RSUs can provide much higher transmitting power compared to vehicles, we have ${\Gamma _{\text{V2V}}} > {\Gamma _{\text{V2V}}} $, i.e., $R_{\text{V2I}}>R_\text{V2V}$.}

\subsection{{Overview of Semantic-aware PreCMTS}}
\label{Dynamics}
As shown in Fig.~\ref{systemmodel}, the proposed PreCMTS consists of {\color{Blues}three} stages, {\color{Blues}namely {preparation for PreCMTS}, {development for PreCMTS}, and {execution  for PreCMTS}, respectively.}

In Stage~1, the vehicles and the MEC server exchange the necessary information for the strategy design. Specifically, when the target vehicle sends a task request to the MEC server via RSU~A, it sends the information about its current position and  average speed to the CC at the same time. Then, the CC broadcasts the request for cooperative transmission  to the vehicles in the coverage of RSU B.  After receiving the request, the relay candidates {\color{Blues}of the vehicles driving in the opposite direction  to the target vehicle} send their current position and average speed back to the CC.

In Stage~2, according to the information reported by the vehicles, the CC first predicts the key parameters about the vehicle trajectories. Then, based on the above predictable parameters, the CC derives the achievable throughput ${Q_{\max }}$ within  ${T_{\max }}$, which is  fed to the SE module to extract an appropriate SR. {\color{Blues}Then, we denote the computational latency for the high-dimensional optimization problem of the cooperative transmission strategy by $\bar t$~\footnote{\color{Blues}The value of $\bar t$ is jointly determined by the server computational capability and the expected performance gain of PreCMTS. It is to be noted that  the vehicle driving out of  RSU B's coverage  within $\bar t$ will not be able to act as a relay.}. Considering the high dynamic nature of the vehicle network, the CC first predicts the locations of the relay candidates at time $\bar t$ to ensure a well-matched PreCMTS to the practical world.  To mitigate the impact of computing latency on the transmission delay, the SUs with large $\alpha $ can be transmitted to the target vehicle in advance via V2I link upon the determination of SR, before the completion of the algorithm execution. We denote the end moment of SU's advance transmission by $\hat t$. Considering the indivisibility of SUs, the   CC needs to predict the location of the target vehicle at time $\max \left\{ {\bar t,\hat t} \right\}$.  }
Next, for the given SR, the CC develops a cooperative transmission strategy ${{\bf{\Phi }}^*}$ based on the predicted locations, which  jointly determines the {\color{Blues}mutually coupled} relay selection and SU assignment.

{\color{Blues}In Stage 3,}  the transmission process begins. It can be further divided into two phases.
 In phase~1, according to ${{\bf{\Phi }}^*}$, the SUs assigned to the direct V2I link are transmitted to the target vehicle directly via RSU A {\color{Blues}as scheduled}. At the same time, the other SUs  are transmitted simultaneously to the corresponding relay vehicles in advance via RSU~B.  When a relay  receives all the SUs assigned to it or it leaves the coverage of RSU~B, the corresponding V2I link is disconnected. Then, in phase~2, the relays forward pre-stored SUs to the target vehicle in order of encounter sequence.  It is to be noted that the relay selection is performed proactively in Stage 1. In this phase, only the V2V communication following strategy ${{\bf{\Phi }}^*}$ happens, and neither the target vehicle nor the CC needs to perform further relay selection. 
\section{ Predictive cooperative multi-relay transmission Strategy }\label{sec:optimization}
\subsection{{Preliminary}}
\label{sec:P}
In this section, we first introduce the predictive parameters used to develop the PreCMTS, i.e.,  the residual dwell time of the vehicles in their associated RSUs, as well as the encounter time and
V2V link lifetime of each relay with the target vehicle.  
{\color{Blues}To concise the notation system and without loss of the generality of the following analysis, we assume the computing latency $\bar t = 0$, that is, treating the moment when the target vehicle sends the request as the initial moment of the transmission process, and thus $\hat t = \bar t$}
For ease of illustration, we take  an example of a two-vehicle encounter process as shown in Fig.~\ref{meeting_diagram}.
 We denote the initial distance to the connected RSU and position of vehicle ${v_i} \in {\cal V}$ by ${d_i }$ and  $\ell_i $, respectively, {\color{Blues}where $0 < \left| {{l_i}} \right| = {d_i} \le {r_{\rm{I}}}$}. If the offset of the initial distance is in the same direction as the vehicle drives, $\ell_i=d_i$, and otherwise, $\ell_i=-d_i$. We denote the  average speed of  vehicle ${v_i} \in {\cal V}$ by ${{\bar u}_i}$ {\color{Blues}and the relative speed of the two vehicles by $\hat u_i = {{\bar u_i}} + {{\bar u_0}}$.} {\color{Blues}Considering the fact that, in practice, the road length is much larger than the road width and the RSU height, we ignore the road width and the RSU length~\cite{wang2016cooperative}.
As such, the communication distance of a V2I and V2V link at time $t$ can be expressed by $d_i^{\rm{I}}\left( t \right) = \left| {{l_i} + {{\bar u}_i}t} \right|$ and $d_i^{\rm{V}}\left( t \right) = \left| {{r_{\rm{V}}} - {{\hat u}_i}\left( {t - \Delta _i^{\rm{T}}} \right)} \right|$, respectively. }
Meanwhile, the residual dwell time of the vehicles in their associated RSUs can be predicted by
\begin{equation}
    D_i^{\rm{I}} = \frac{{{r_{\rm{I}}} - {\ell _i}}}{{{{\bar u}_i}}},i \in \left\{ {0,1, \ldots ,\left| {{{\cal V}_{\rm{R}}}} \right|} \right\}. 
\end{equation}
We refer to  driving into the communication range of the target vehicle as an encounter with the target vehicle. The encounter time between relay ${v_i}$ and  vehicle ${v_0}$ can be predicted by
\begin{equation}
   \Delta _i^{\rm{T}} = \frac{{H - {\ell _0} - {\ell _i} - {r_{\rm{V}}}}}{{{{\hat u}_i}}},i \in \left\{ {1, \ldots ,\left| {{\mathcal{V}_{\rm{R}}}} \right|} \right\}.
\end{equation}
Moreover, the duration of relay ${v_i}$ within the coverage of vehicle $v_0$ can be predicted by 
\begin{equation}
     D_i^{\rm{T}} = \frac{{2{r_{\rm{V}}}}}{{{{\hat u}_i}}},i \in \left\{ {1, \ldots ,\left| {{\mathcal{V}_{\rm{R}}}} \right|} \right\}
\end{equation}
Based on the above predictable parameters,  the cumulative data amount transmitted via each V2I link or V2V link (without considering the existence of other vehicles communicating), can be calculated by $ C_i^{\rm{I}} = {R_{\rm{I}}}D_i^{\rm{I}},\left( {i \in \left\{ {0,1, \ldots ,\left| {{\mathcal{V}_{\rm{R}}}} \right|} \right\}} \right)$ and $ C_i^{\rm{V}} = {R_{\rm{V}}}D_{{i}}^{\rm{T}},\left( {i \in \left\{ {1, \ldots ,\left| {{\mathcal{V}_{\rm{R}}}} \right|} \right\}} \right)$, respectively. {\color{Blues}Meanwhile, considering the mutual independence of large- and small-scale channel gain, the average overall energy consumption for a link can be calculated by 
\begin{equation}
\begin{aligned}
    f_i^x\left( {{t^{\rm{S}}},{t^{\rm{E}}}} \right) &= \int_{{t^{\rm{S}}}}^{{t^{\rm{E}}}} {\int_0^\infty  {\frac{{{\Upsilon _x}{\sigma ^2}}}{{{A_x}{G_x}\left( {d_i^{x}\left( t \right)} \right)\tilde g}}f\left( {\tilde g} \right)d\tilde gdt} }  \\ &= \int_{{t^{\rm{S}}}}^{{t^{\rm{E}}}} {{\mathbb E _{\tilde g}}\left[ {{p_x}\left( {d_i^x\left( t \right)} \right)} \right]} dt, \label{ppp}  
\end{aligned}
\end{equation}
where $t^\text{S}$ and $t^\text{E}$ represent the start and end time of a link, respectively.
By substituting \eqref{ap} into \eqref{ppp}, the final expression of \eqref{ppp} can be obtained, as shown in \eqref{P_V2I} and \eqref{P_V2V}. }
\begin{figure*}[bp]
\vspace{-0.5cm}
\centering
\hrulefill
\textcolor{Blues}{
   \begin{equation}
 f_i^{{\rm{V2I}}}\left( {{t^{\rm{S}}},{t^{\rm{E}}}} \right) = \left\{ {\begin{array}{*{20}{c}}
{\frac{{{\rm{M}}{\Upsilon _{\rm{I}}}{\sigma ^2}}}{{{A_{\rm{I}}}{b_{\rm{I}}}{{\bar u}_i}\left( {{a_{\rm{I}}} + 1} \right)}}\left( {{{\left( {{\ell _i} + {{\bar u}_i}{t^{\rm{E}}}} \right)}^{{a_{\rm{I}}} + 1}} - {{\left( {{\ell _i} + {{\bar u}_i}{t^{\rm{S}}}} \right)}^{{a_{\rm{I}}} + 1}}} \right),}&{\ell  \ge 0,0 \le {t^{\rm{S}}} \le {t^{\rm{E}}} \le D_i^{\rm{I}}}\\
{\frac{{{\rm{M}}{\Upsilon _{\rm{I}}}{\sigma ^2}}}{{{A_{\rm{I}}}{b_{\rm{I}}}{{\bar u}_i}\left( {{a_{\rm{I}}} + 1} \right)}}\left( {{{\left( { - {\ell _i} - {{\bar u}_i}{t^{\rm{S}}}} \right)}^{{a_{\rm{I}}} + 1}} - {{\left( { - {\ell _i} - {{\bar u}_i}{t^{\rm{E}}}} \right)}^{{a_{\rm{I}}} + 1}}} \right),}&{\ell  < 0,0 \le {t^{\rm{S}}} \le {t^{\rm{E}}} \le  - \frac{{{\ell _i}}}{{{{\bar u}_i}}}}\\
{\frac{{{\rm{M}}{\Upsilon _{\rm{I}}}{\sigma ^2}}}{{{A_{\rm{I}}}{b_{\rm{I}}}{{\bar u}_i}\left( {{a_{\rm{I}}} + 1} \right)}}\left( {{{\left( { - {\ell _i} - {{\bar u}_i}{t^{\rm{S}}}} \right)}^{{a_{\rm{I}}} + 1}} + {{\left( {{\ell _i} + {{\bar u}_i}{t^{\rm{E}}}} \right)}^{{a_{\rm{I}}} + 1}}} \right),}&{\ell  < 0,0 \le {t^{\rm{S}}} - \frac{{{\ell _i}}}{{{{\bar u}_i}}} \le {t^{\rm{E}}} \le D_i^{\rm{I}}}\\
{\frac{{{\rm{M}}{\Upsilon _{\rm{I}}}{\sigma ^2}}}{{{A_{\rm{I}}}{b_{\rm{I}}}{{\bar u}_i}\left( {{a_{\rm{I}}} + 1} \right)}}\left( {{{\left( {{\ell _i} + {{\bar u}_i}{t^{\rm{E}}}} \right)}^{{a_{\rm{I}}} + 1}} - {{\left( {{\ell _i} + {{\bar u}_i}{t^{\rm{S}}}} \right)}^{{a_{\rm{I}}} + 1}}} \right),}&{\ell  < 0, - \frac{{{\ell _i}}}{{{{\bar u}_i}}} \le {t^{\rm{S}}} \le {t^{\rm{E}}} \le D_i^{\rm{I}}}
\end{array}} \right.\label{P_V2I}
\end{equation}
\begin{equation}
 f_i^{{\rm{V2V}}}\left( {{t^{\rm{S}}},{t^{\rm{E}}}} \right) = \left\{ {\begin{array}{*{20}{l}}
{\frac{{{\rm{M}}{\Upsilon _{\rm{I}}}{\sigma ^2}}}{{{A_{\rm{I}}}{b_{\rm{V}}}{{\hat u}_i}\left( {{a_{\rm{V}}} + 1} \right)}}\left( {{{\left( {{r_{\rm{V}}} - {{\hat u}_i}\left( {{t^{\rm{S}}} - \Delta _i^{\rm{T}}} \right)} \right)}^{{a_{\rm{V}}} + 1}} - {{\left( {{r_{\rm{V}}} - {{\hat u}_i}\left( {{t^{\rm{E}}} - \Delta _i^{\rm{T}}} \right)} \right)}^{{a_{\rm{V}}} + 1}}} \right),}&{0 \le {t^{\rm{S}}} \le {t^{\rm{E}}} \le \frac{{D_i^{\rm{T}}}}{2}}\\
{\frac{{{\rm{M}}{\Upsilon _{\rm{I}}}{\sigma ^2}}}{{{A_{\rm{I}}}{b_{\rm{V}}}{{\hat u}_i}\left( {{a_{\rm{V}}} + 1} \right)}}\left( {{{\left( {{r_{\rm{V}}} - {{\hat u}_i}\left( {{t^{\rm{S}}} - \Delta _i^{\rm{T}}} \right)} \right)}^{{a_{\rm{V}}} + 1}} + {{\left( {{{\hat u}_i}\left( {{t^{\rm{E}}} - \Delta _i^{\rm{T}}} \right) - {r_{\rm{V}}}} \right)}^{{a_{\rm{V}}} + 1}}} \right),}&{0 \le {t^{\rm{S}}} \le \frac{{D_i^{\rm{T}}}}{2} \le {t^{\rm{E}}} \le D_i^{\rm{T}}}\\
{\frac{{{\rm{M}}{\Upsilon _{\rm{I}}}{\sigma ^2}}}{{{A_{\rm{I}}}{b_{\rm{V}}}{{\hat u}_i}\left( {{a_{\rm{V}}} + 1} \right)}}\left( {{{\left( {{{\hat u}_i}\left( {{t^{\rm{E}}} - \Delta _i^{\rm{T}}} \right) - {r_{\rm{V}}}} \right)}^{{a_{\rm{V}}} + 1}} - {{\left( {{{\hat u}_i}\left( {{t^{\rm{S}}} - \Delta _i^{\rm{T}}} \right) - {r_{\rm{V}}}} \right)}^{{a_{\rm{V}}} + 1}}} \right),}&{\frac{{D_i^{\rm{T}}}}{2} \le {t^{\rm{S}}} \le {t^{\rm{E}}} \le D_i^{\rm{T}}}
\end{array}} \right.    \label{P_V2V} 
\end{equation}}

\end{figure*}

\subsection{{Achievable Throughput Analysis}}
\label{ATA}

As this subsection is to derive the achievable throughput within $T_{\text max}$, the issue of energy saving is not considered here. With the consideration of $R_{\text{I}}>R_\text{V}$, we let the target vehicle maintain the V2I link until it leaves the coverage of RSU A. Moreover, since all the V2V links can provide the same average transmission rate $R_\text{V}$, we transform the problem of deriving the maximum achievable throughput into the problem of deriving the maximum total duration of the V2V links established sequentially between the target vehicle and the relays.

As mentioned in Section~\ref{Description}, we number the relays in the sequence of encounters with the target vehicle. In this context, we have $\Delta _1^{\text{T}} \leqslant \Delta _2^{\text{T}} \leqslant  \ldots  \leqslant \Delta _{\left| {{\mathcal{V}_{\text{R}}}} \right|}^{\text{T}}$.
We characterize a communication link by an interval whose two endpoints represent the start time and the end time  of the link. For example, without considering the effect of other existing communication links, an available V2V link  of relay $v_i$ can represented by $\left[ {\Delta _i^{\text{T}},\Delta _i^{\text{T}} + D_i^{\text{T}}} \right]$.
However, for a specific cooperative transmission, the start time and end time of of each V2V link become less straightforward. To avoid the interference between V2V links,  relay $v_i$ can only communicate to  vehicle $v_0$ after the V2V link between vehicle $v_0$ and relay $v_{i-1}$ is broken. This means that the start time of the V2V link established by vehicle $v_i$ may be later than $\Delta_i^{\text T}$.
Meanwhile,  relay $v_i$ is responsible for forwarding the pre-stored SUs to vehicle $v_0$. Thus, the maximum cumulatively transmitted data amount (denoted by $\hat C_i^{{\text{I}}}$) of the previously maintained V2I link restricts the maximum data amount transmitted over the V2V link. This requires a more elaborate calculation for the end time of the V2V link.

We denote the time when relay $v_i$ starts to forward data to vehicle $v_0$ in a cooperative transmission achieving the maximum throughput by $\hat \delta _i$. Meanwhile,  the set of the total V2V links cumulatively established by vehicle $v_0$ when  it encounters relay $v_i$ is denoted by ${\mathcal{D}_i}$.  For example, for relay $v_1$, the V2V link needs be established after the target vehicle leaves the coverage area of RSU A. Therefore, we have
\begin{equation}
\hat \delta _1{\text{  =  max}}\left\{ {\Delta _1^{\text{T}},0 + D_{\text{0}}^{\text{I}}} \right\}. 
\end{equation}
Since the end time of a V2V link is restricted by two constraints: the maximum data amount pre-stored via the V2I link, and the time to leave the communication range of  the target vehicle,  the V2V link established by relay $v_1$, i.e., ${\mathcal{D}_1}$, is expressed by
\begin{equation}
    {\mathcal{D}_1} = \left[ {\hat \delta _1,\min \left( {\hat \delta _1 + {{\hat C_1^{\text{I}}} \mathord{\left/
 {\vphantom {{\hat C_i^{\text{I}}} {{R_{\text{V}}}}}} \right.
 \kern-\nulldelimiterspace} {{R_{\text{V}}}}},\Delta _1^{\text{T}} + D_{\text{1}}^{\text{T}}} \right)} \right], 
\end{equation}
{\color{Blues}where $\hat C_1^{\rm{I}} = {R_{\rm{I}}} \cdot \min \left\{ {D_1^{\rm{I}},\left| {\left[ {0,{{\hat \delta }_1}} \right]} \right|} \right\}$, since the pre-store process for the relay vehicle needs to be completed before forwarding data to the target vehicle.}
Moreover, as shown in Fig.~\ref{max_V2V_duration}, for  subsequent relay $v_i$, $i \in \left\{ {2, \ldots ,\left| {{{\cal V}_{\text{R}}}} \right|} \right\}$, the start time of the V2V link should be after the end time of the previous V2V link established by relay $v_{i-1}$. Therefore, the start time of the subsequent V2V links is expressed by  
\begin{equation}
    \hat \delta _i = \Delta _i^{\text{T}} + \left| {{\mathcal{D}_{i - 1}}\bigcap {\left[ {\Delta _i^{\text{T}},\Delta _i^{\text{T}} + D_i^{\text{T}}} \right]} } \right|, i \in \left\{ {2, \ldots ,\left| {{\mathcal{V}_{\text{R}}}} \right|} \right\},  
\end{equation}
where $\left|  \cdot  \right|$ represents the interval length, i.e., the link duration. Similar to the derivation of ${\mathcal{D}_1}$, {\color{Blues}we have $\hat C_i^{\rm{I}} = {R_{\rm{I}}} \cdot \min \left\{ {D_i^{\rm{I}},\left| {\left[ {0,{{\hat \delta }_i}} \right]} \right|} \right\}$.} Then, the set of the total V2V links cumulatively established when vehicle $v_0$ encountering $v_i$, $i \in \left\{ {2, \ldots ,\left| {{\mathcal{V}_{\text{R}}}} \right|} \right\}$, is represented by
\begin{equation}
   {\mathcal{D}_i} = \left[ {\hat \delta _i,\min \left( {\hat \delta _i + {{\hat C_i^{\text{I}}} \mathord{\left/
 {\vphantom {{\hat C_i^{\text{I}}} {{R_{\text{V}}}}}} \right.
 \kern-\nulldelimiterspace} {{R_{\text{V}}}}},\Delta _i^{\text{T}} + D_i^{\text{T}}} \right)} \right] \cup {\mathcal{D}_{i - 1}}. \label{V2V}
\end{equation}

According to~\eqref{V2V}, we can obtain the maximum total duration of the V2V links, i.e., ${\mathcal{D}_{\left| {{\mathcal{V}_{\text{R}}}} \right|}}$.
Furthermore, considering the delay requirement, the total duration of the V2V links is further modified to ${\mathcal{D}_{\left| {{\mathcal{V}_{\text{R}}}} \right|}} \cap \left[ {0,{T_{\max }}} \right]$. Therefore, by jointly considering the direct V2I link between vehicle $v_0$ and RSU~A, the achievable throughput within $T_{\text max}$ can be expressed by 
\begin{equation}
{Q_{\max }} = {R_{\text{I}}}D_0^{\text{I}} + {R_{\text{V}}}\left| {{\mathcal{D}_{\left| {{\mathcal{V}_{\text R}}} \right|}}\bigcap {\left[ {0,{d^{\max }}} \right] - \left[ {0,D_0^{\text{I}}} \right]} } \right|. \label{Qmax}
\end{equation}

\subsection{{Problem Formulation for Relay Selection \& SU Assignment}}
\label{OPF}
According to $Q_\text{max}$ obtained in Section~\ref{ATA}, the semantic encoder extracts an  appropriate SR.
We assume that the SR consists of $N$ SUs, i.e., $j\in \left\{ {1,2, \ldots ,N} \right\}$. 
{\color{Blues}Considering that
 the direct link relies on the least predictive parameters and its start and end time is independent of other V2V links, the sudden change in vehicle speed have minimal impact on its transmission integrity. In this sense, the SUs with high importance need to prioritize the direct link, which ultimately determine the amount of total data to be transmitted via forwarding links. Therefore, relay selection and SU allocation are two mutually coupled problems, which are jointly characterized by a  (${\left| {{\mathcal{V}_{\text{R}}}} \right| + 1}$)-row and $N$-column matrix ${{\mathbf{\Phi }}} = ({\phi _{i,j}}:i \in \left\{ {0,1,2, \ldots ,\left| {{\mathcal{V}_{\text R}}} \right|} \right\},j \in \left\{ {1,2, \ldots ,N} \right\})$. } Herein, ${\phi _{i,j}}$ is a binary
indicator, with ${\phi _{i,j}}=1$ meaning that SU $j$ is transmitted to vehicle $i$ via the direct link or the relay link, and  ${\phi _{i,j}}=0$ otherwise. If  $\sum\nolimits_{j = 1}^N {{\phi _{i,j}}}  = 0$, it means that vehicle ${v_i} \in {\mathcal{V}_{\text{R}}}$ is not selected as a  relay under ${\mathbf{\Phi }}$.
Before defining the optimal strategy ${\mathbf{\Phi }}^*$, we first analyze the constraints that a feasible policy needs to satisfy as follows.

For a certain ${\mathbf{\Phi }}$, the start and the end time of each links are deterministic. As stated in Section~\ref{sec:P}, we assume that all the V2I links are established at the initial moment in the theoretical study of this paper. Thus, the end time of the V2I links are only determined by the SUs assigned to the each vehicle. As such, the end time of the V2I link of vehicle ${v_i} \in \mathcal{V}$ is expressed by $t_i^{{{\text{E}}_{\text{I}}}} = \frac{{\sum\nolimits_{j = 1}^N {{\phi _{i,j}}{\beta _j}} }}{{{R_{\text{I}}}}}$. Moreover, we denote the start time of the V2V link established by vehicle $v_i$ by ${t_i^{\text{S}_{\text{V}}}}$, ($\Delta _i^{\text{T}} \leqslant {t_i^{{{\text{S}}_{\text{V}}}}} \leqslant \Delta _i^{\text{T}} + D_i^{\text{T}}$).  For vehicle $v_1$, the start time of the V2V link should be after the end time of the V2I link between the target vehicle $v_0$ and RSU A.  Thus, the expression of $t_1^{{{\text{S}}_{\text{V}}}}$ is shown as
\begin{equation}
   t_1^{{{\text{S}}_{\text{V}}}} = \min \left\{ {\max \left\{ {\Delta _1^{\text{T}},t_0^{{{\text{E}}_{\text{I}}}}} \right\},\Delta _1^{\text{T}} + D_1^{\text{T}}} \right\}.  
\end{equation}
Similarly, all the subsequent V2V links should start after all their previous links are broken. We denote the end time of the V2V link established by vehicle ${v_i} \in {\mathcal{V}_{\text{R}}}$ by $t_i^{\text{E}_{\text{V}}}$, which is calculated by $t_i^{\text{E}_\text{V}} = t_{i}^{{{\text{S}}_{\text{V}}}}+ \frac{{\sum\nolimits_{j = 1}^N {{\phi _{i,j}}{\beta _j}} }}{{{R_{\text{V}}}}}$. Therefore, the start time of the subsequent V2V links is expressed by 
\begin{equation}
   t_i^{{{\text{S}}_{\text{V}}}} = \min \left\{ {\max \left\{ {\Delta _i^{\text{T}},t_{i - 1}^{{{\text{E}}_{\text{V}}}}} \right\},\Delta _i^{\text{T}} + D_i^{\text{T}}} \right\},\forall i \in \left\{ {2, \ldots ,\left| {{\mathcal{V}_{\text{R}}}} \right|} \right\}.  
\end{equation}
After determining the start time, the maximum data amount can be transmitted a V2V link can by calculated by $\hat{C}_i^{\text{V}}={{R_{\text{V}}}\left( {\Delta _i^{\text{T}} + D_i^{\text{T}} - {t_i^{{{\text{S}}_{\text{V}}}}}} \right)}$, {\color{Blues}and the  maximum pre-stored data amount via the V2I link can be calculated by $\hat C_i^{\rm{I}} = {R_{\rm{I}}} \cdot \min \left\{ {D_i^{\rm{I}},\left| {\left[ {0,t_i^{{{\rm{S}}_{\rm{V}}}}} \right]} \right|} \right\}$}. Since the total data amount of the SUs assigned is bounded by both the transmission capacity of the V2V link and V2I link, to ensure the integrity of the transmission, we have 
\begin{equation}
    \sum\nolimits_{j = 1}^N {{\phi _{i,j}}{\beta _j}}  \leqslant \min \left\{ {\hat C_i^{\text{I}},\hat C_i^{\text{V}}} \right\},{\text{ }}\forall i \in \left\{ {1, \ldots ,\left| {{\mathcal{V}_{\text R}}} \right|} \right\}. \label{C2}
\end{equation}
For the same reason,  the transmission of the SUs assigned to vehicle $v_0$ is required to be completed within the coverage of RSU A. Thus, we have 
\begin{equation}
   \sum\nolimits_{j = 1}^N {{\phi _{0,j}}{\beta _j}}  \leqslant \hat C_0^{\text{I}}. \label{C1}
\end{equation}
Additionally, considering the delay requirement,  the last V2V link should end at a time earlier than the maximum acceptable delay threshold. Therefore, we have
\begin{equation}
   \max \left\{ {t_{i}^{{{\text{E}}_{\text{V}}}}} \right\} \leqslant {T_{\max }},{\text{ }}\forall i \in \left\{ {1, \ldots ,\left| {{\mathcal{V}_{\text{R}}}} \right|} \right\}. 
\end{equation}

To evaluate feasible strategies that satisfy the above constraints, we consider two main aspects. One is the energy consumption.
According to \eqref{P_V2I} and \eqref{P_V2V}, for any feasible strategy ${\mathbf{\Phi }}$, the total energy consumption of the V2I links and V2V links can be calculated by ${P_{{\text{V2I}}}} = \sum\nolimits_{i = 0}^{\left| {{\mathcal{V}_{\text{R}}}} \right|} {f_i^{{\text{V2I}}}\left( {0,t_i^{{{\text{E}}_{\text{I}}}}} \right)} $ and ${P_{{\text{V2V}}}} = \sum\nolimits_{i = 1}^{\left| {{\mathcal{V}_{\text{R}}}} \right|} {f_i^{{\text{V2V}}}( {t_i^{{{\text{S}}_{\text{V}}}},t_i^{{{\text{E}}_{\text{V}}}}} )}$, respectively. 
The other is  semantic reliability. 
Considering the possibility of sudden changes in vehicle speed, the selected SR might fail to be fully transmitted as planned. 
 Therefore, the SUs with high importance can assign to the more reliable direct link to reduce the impact of vehicle network uncertainty as discussed at the beginning of this subsection.
With this in mind,  we introduce two parameters ${ \theta _{\text{T}}}$ and ${\theta _{\text{R}}}$ to qualitatively characterize the reliability of the direct and forward transmission, respectively, where ${\theta _{\text{T}}} > {\theta _{\text{R}}}$. Furthermore, we define a new metric to quantify  the semantic reliability of ${\mathbf{\Phi }}$ based on semantic significance ${\alpha _j}$ of each SU $j$, which is expressed by  
\begin{equation}
 \Theta {\text{ = }}{\theta _{\text{T}}}\sum\nolimits_{j = 1}^N {{\phi _{0,j}}{\alpha _j}}  + {\theta _{\text{R}}}\sum\nolimits_{i = 1}^{\left| {{\mathcal{V}_{\text{R}}}} \right|} {\sum\nolimits_{j = 1}^N {{\phi _{i,j}}{\alpha _j}}}. \label{DDD}   
\end{equation}

In summary, the relay selection and SU assignment can be jointly formulated as a combinatorial optimization problem,
\begin{equation}
    \mathop {\min }\limits_{\mathbf{\Phi }} \kappa_1(P_{{\text{V2I}}}  + P_{{\text{V2V}}}) - \kappa_2 \Theta, \label{15}\tag{P1}
\end{equation}
subject to 
\begin{equation}
  {\phi _{i,j}} \in \left\{ {0,1} \right\},{\text{                          }}\forall i \in \left\{ {0, \ldots ,\left| {{\mathcal{V}_\text{R}}} \right|} \right\},\forall j \in \left\{ {1, \ldots ,N} \right\},\tag{a}\label{15a}
\end{equation}
\begin{equation}
    \sum\nolimits_{i = 0}^{\left| {{\mathcal{V}_\text{R}}} \right|} {{\phi _{i,j}}}  = 1,{\text{                                                     }}\forall j \in \left\{ {1,2, \ldots ,N} \right\},\tag{b}\label{15b}
\end{equation}
\begin{equation}
    {\text{Constraints: (21)\quad (22)\quad (23)}}, \nonumber
\end{equation}

\noindent where  $\kappa_1$ and $\kappa_2$ are two  parameters used to weigh {\color{Blues} energy consumption and semantic transmission reliability}. Moreover, the constraints in \eqref{15a} and \eqref{15b} ensure that each SU is assigned only once. The specific solution to \eqref{15} is provided in Section~\ref{MAS}.

\subsection{Markov Approximation and Solution}
\label{MAS}
Given the multiple $\max \left\{  \cdot  \right\}$ and $\min \left\{  \cdot  \right\}$ in~\eqref{15}, the explicit expression for its feasible region is challenging to derive. Also, due to the high dimensionality of ${\mathbf{\Phi }}$, the conventional numerical analysis methods and centralized search algorithms become inefficient, {\color{Blues}especially for finding a favourable solution within limited time}. 

To this end, we propose a Markov-chain-guided multi-thread search algorithm (M-MTSA) as shown in Fig.~\ref{algorithm2}. 
Inspired by {Markov approximation}~\cite{chen2013markov}, we first approximate \eqref{15} by transforming it into a continuous convex optimization problem \eqref{17} in the probability domain based on Log-Sum-Exp approximation and the conjugate function property. The decision variables in \eqref{17} are the probability weights corresponding to all possible  ${\mathbf{\Phi }}$. Ideally, the probability corresponding to the optimal strategy ${\mathbf{\Phi }}^*$ is remarkably close to one. Then, we construct a Markov chain with the state space as all possible ${\mathbf{\Phi }}$ and the stationary distribution as the optimal solution of \eqref{17}.  During its execution, according to constraint~\eqref{15b}, M-MTSA maintains $N$ threads for all the SUs, respectively.  The partial strategy for thread $j$ serves as the $j$th column of ${\mathbf{\Phi }}$, a \textit{one-hot} vector,  which determines the selected relay for SU $j$. According to the transition rates of the Markov chain,  the individual relays for the SUs are constantly and distributively updated  with  small inter-thread message passing overhead. By the careful design of transition rates, the Markov chain jumps to better strategies over time.
Next, we detail the problem transformation, Markov chain construction, and M-MTSA design in Sections~\ref{sec:transformation}-\ref{Design}, respectively.

\subsubsection{Problem Transformation}
\label{sec:transformation}
Recall the problem in \eqref{15}, for ease of presentation, we denote the objective function by  $U\left( {\mathbf{\Phi }} \right)$. Then, \eqref{15} can be rewritten as
\begin{equation}
    \mathop {\min  }\limits_{{\mathbf{\Phi }} \in \mathcal{F^*}} U\left( {\mathbf{\Phi }} \right),\label{16} 
\end{equation}
where $\mathcal{F^*}$ represents the feasible region. Since $\mathcal{F^*}$ is unavailable, we transform constrained problem to unconstrained one by adding a penalty term to the objective function. Then, \eqref{15} can be rewritten as
\begin{equation}
    \mathop {\min }\limits_{{\mathbf{\Phi }} \in \mathcal{F}} U\left( {\mathbf{\Phi }} \right) + \Omega  \cdot {{\mathbf{1}}_{{\complement _\mathcal{F}}{\mathcal{F}^*}}}\left( {\mathbf{\Phi }} \right).\label{17}\tag{P2}
\end{equation}
In \eqref{17}, $\mathcal{F}$ is the set of all the possible ${\mathbf{\Phi }}$ satisfying constraints \eqref{15a} and \eqref{15b} with $\left| \mathcal{F} \right|= {( {\left| {{\mathcal{V}_{\text{R}}}} \right| + 1} )^N}$. Moreover,  $\Omega$ is a constant penalty factor which is significantly larger than $U\left( {\mathbf{\Phi }} \right)$, and ${{\mathbf{1}}_{{\complement _\mathcal{F}}{\mathcal{F}^*}}}\left( {\mathbf{\Phi }} \right)$ is an indicator function defined as
\begin{equation}
    {{\mathbf{1}}_{{\complement _\mathcal{F}}{\mathcal{F}^*}}}\left( {\mathbf{\Phi }} \right) = \left\{ {\begin{array}{*{20}{c}}
  {1,}&{{\mathbf{\Phi }} \in {\complement _\mathcal{F}}{\mathcal{F}^*}} \\ 
  {0,}&{{\text{otherwise}}} 
\end{array}} \right..  
\end{equation}
For brevity, we rewrite the objective function in \eqref{17} as $\hat U\left( {\mathbf{\Phi }} \right)$. To enable the analysis from the probability domain, the log-sum-exp function is used to approximate ${\min _{{\mathbf{\Phi }} \in \mathcal{F}}}\hat U\left( {\mathbf{\Phi }} \right)$, i.e.,
%employed and an approximate expression for \eqref{19} is given as
\begin{equation}
  \mathop {\min }\limits_{{\mathbf{\Phi }} \in \mathcal{F}} \hat U\left( {\mathbf{\Phi }} \right) \approx   {g_\varpi }( \hat{ \mathcal{U}} ) =  - \varpi \log \left( {\sum\limits_{{\mathbf{\Phi }} \in \mathcal{F}} {\exp \left( { - \frac{{\hat U\left( {\mathbf{\Phi }} \right)}}{\varpi }} \right)} } \right),    
\end{equation}
with the upper bound of $\varpi \log \left| \hat{ \mathcal{U}} \right|$ for approximation gap.
\begin{lemma}
\label{L1}
When $\varpi  \to {0^ + }$, for a set $\mathcal{X}$ of $n$ nonnegative real variables $x_1$, $x_2$, $x_3$, ..., $x_n$, we have  
\begin{equation}
   \mathop {\min }\limits_{i = 1,2, \ldots ,n} {x_i} - \varpi \log \left| \mathcal{X} \right| \leqslant {g_\varpi }\left( \mathcal{X} \right) \leqslant \mathop {\min }\limits_{i = 1,2, \ldots ,n} {x_i}.  
\end{equation}
\end{lemma}
\begin{proof}
We  rearrange $x_i$ so that they are ranked as ${x_1} \leqslant {x_2} \leqslant  \ldots  \leqslant {x_n}$. Then, we have 
\begin{equation}
\begin{aligned}
  {g_\varpi }\left( \mathcal{X} \right) &=- \varpi \log \left( {\sum\nolimits_{i = 1}^n {\exp \left( { - \frac{{{x_i}}}{\varpi }} \right)} } \right) \\
   &=  - \varpi \log \left( {\exp \left( { - \frac{{{x_1}}}{\varpi }} \right)\exp \left( {\frac{{{x_1}}}{\varpi }} \right)\sum\nolimits_{i = 2}^n {\exp \left( { - \frac{{{x_i}}}{\varpi }} \right)} } \right) \\
   &={x_1} - \varpi \log \left( {1 + \sum\nolimits_{i = 2}^n {\exp \left( {\frac{{{x_1} - {x_i}}}{\varpi }} \right)} } \right).    
\end{aligned}
\end{equation}
Therefore, the approximation gap can be expresses by
\begin{equation}
    \left| {{g_\varpi }\left( \mathcal{X} \right) - {x_1}} \right| = \left| {\varpi \log \left( {1 + \sum\limits_{i = 2}^n {\exp \left( {\frac{{{x_1} - {x_i}}}{\varpi }} \right)} } \right)} \right|.  
\end{equation}
When ${x_1} = {x_2} =  \ldots  = {x_n}$, $ \left| {{g_\varpi }\left( \mathcal{X} \right) - {x_1}} \right|= \varpi \log \left| \mathcal{F} \right|$; when ${x_1} \ll {x_2} \leqslant  \ldots  \leqslant {x_n}$, $\left| {{g_\varpi }\left( \mathcal{X} \right) - {x_1}} \right| \to 0$.
\end{proof}
Then, since ${g_\varpi }( {\hat{\mathcal{U}}} )$ is a convex and closed function, the conjugate of its conjugate is itself, i.e., ${g_\varpi }( \hat{\mathcal{U}}) = g_\varpi ^{**}( {\hat{\mathcal{U}}})$. According to the definition of conjugate function\footnote{Let $g$: ${{\text{R}}^n} \to {\text{R}}$. The conjugate function of $g$ is defined as ${g^*}\left( y \right) = {\sup _{{\mathbf{x}} \in {\text dom} g}}\left( {{y^T}x - g\left( x \right)} \right)$~\cite{boyd2004convex}}, the conjugate of ${g_\varpi }( \hat{\mathcal{U}}) $ can be expressed by~\cite[p.93]{boyd2004convex}
\begin{equation}
    g_\varpi ^*\left( {\mathbf{p}} \right) = \left\{ {\begin{array}{*{20}{c}}
  { - \varpi \sum\nolimits_{{\mathbf{\Phi }} \in \mathcal{F}} {{p_{\mathbf{\Phi }}}\log {p_{\mathbf{\Phi }}}} ,}&{{\text{if }}{\mathbf{p}} \geqslant 0{\text{ and }}{1^T}{\mathbf{p}} = 1;} \\ 
  {\infty ,}&{{\text{otherwise}}{\text{.}}} 
\end{array}} \right.  
\end{equation}
Similarly, the conjugate of $ g_\varpi ^*\left( {\mathbf{p}} \right)$, i.e., $g_\varpi ^{**}( {\hat{\mathcal{U}}})$, can be obtained by solving the following problem~\cite{chen2013markov}.
\begin{equation}
    \begin{gathered}
  \mathop {\max}\limits_{{\mathbf{p}} \geqslant 0} \sum\limits_{{\mathbf{\Phi }} \in \mathcal{F}} {{p_{\mathbf{\Phi }}}} {\hat U}\left( {\mathbf{\Phi }} \right) + \varpi \sum\limits_{{\mathbf{\Phi }} \in \mathcal{F}} {{p_{\mathbf{\Phi }}}} \log {p_{\mathbf{\Phi }}}, \hfill \\
  {\text{s}}{\text{.t}}{\text{.  }}\sum\limits_{{\mathbf{\Phi }} \in \mathcal{F}} {{p_{\mathbf{\Phi }}} = 1.}  \hfill \\ 
\end{gathered} \label{24}\tag{P3}
\end{equation}
Therefore, the optimal value of \eqref{24} is the same as ${g_\varpi }( \hat{\mathcal{U}})$. According to Proposition~\ref{L1}, it approximates the optimal value of \eqref{17} with a gap bounded by $\varpi \log \left| \mathcal{F} \right|$, from the analysis of \eqref{24}, which is caused by the term $\varpi \sum\nolimits_{\bf{\Phi}  \in \mathcal{F}} {{p_{\bf{\Phi}} }\log {p_{\bf{\Phi}} }}$.
 By addressing the Karush–Kuhn–Tucker conditions~\cite{boyd2004convex},  the closed-form of the optimal solution to \eqref{24} is shown as below: 
\begin{equation}
    p_{\mathbf{\Phi }}^* = \frac{{\exp \left( { - \frac{{{\hat U}\left( {\mathbf{\Phi }} \right)}}{\varpi }} \right)}}{{\sum\nolimits_{{\mathbf{\Phi }}' \in \mathcal{F}} {\exp \left( { - \frac{{{\hat U}\left( {{\mathbf{\Phi }}'} \right)}}{\varpi }} \right)} }},\forall {\mathbf{\Phi }} \in \mathcal{F}.\label{25}
\end{equation}
As such, an average performance that is close to the optimal value of \eqref{17}  can be achieved via time-sharing of all the possible $\bf \Phi$ according to individual $p_{\bf \Phi} ^* $. Obviously, according to~\eqref{25}, ${{\mathbf{\Phi }}^*}$ occupies the longest proportion of time. The point to note here is that ${{\mathbf{\Phi }}^*}$ is what we try to find in our work, instead of the average performance itself.

\begin{figure}[t]
 \centering
 \includegraphics[scale = 0.4]{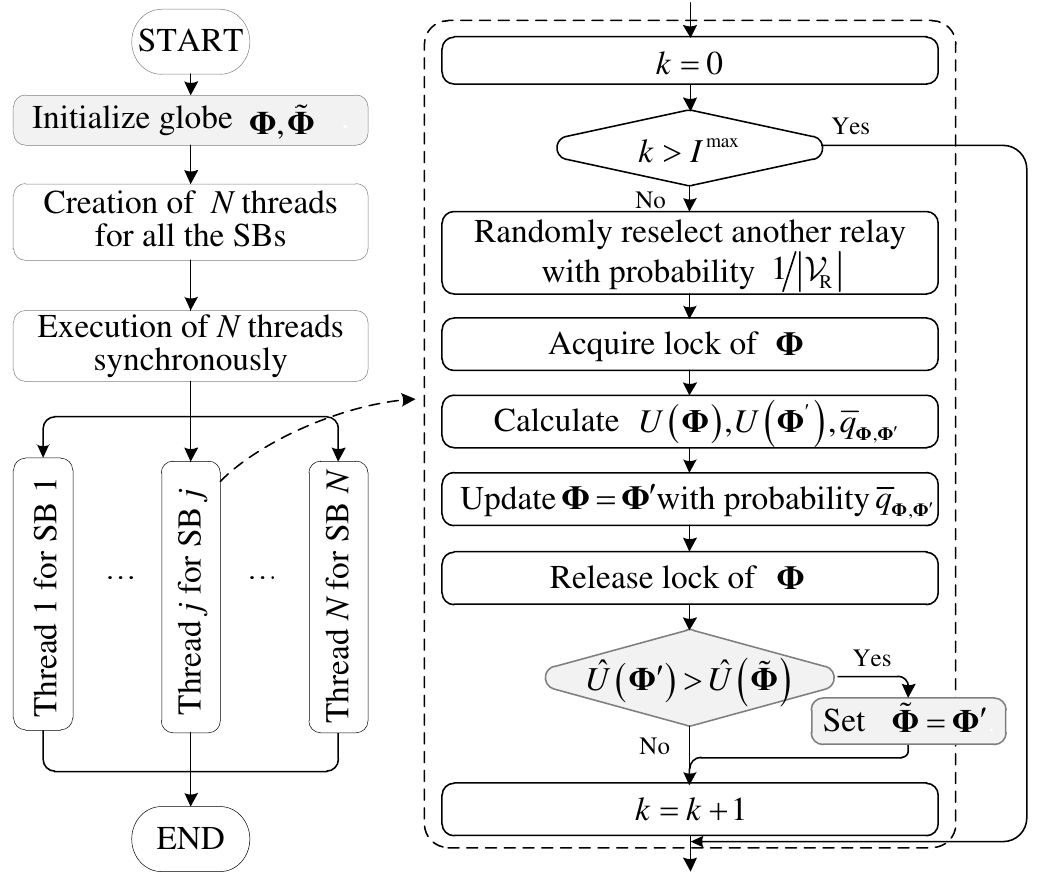}\\
 \caption{The flowchart of M-MTSA.}
 \label{algorithm2}

\end{figure}
\begin{algorithm}
\setstretch{0.3}
\footnotesize
\SetKwInOut{Input}{Input}
        \SetKwInOut{Output}{Output}
        \caption{ SU assignment algorithm based on~\cite{liu2022elastic} (\textbf{Baseline})}
        % \KwIn{this text}
        % \KwOut{how to write algorithm with \LaTeX2e }
        \Input{ ${{R_{{\text{V}}}}}, {{R_{{\text{I}}}}}, D_0^{\text{I}}, D_i^{\text{I}},\Delta _i^{\text{T}},D_i^{{\text{T}}},i \in \left\{ {1,2, \ldots ,\left| {{\mathcal{V}_R}} \right|} \right\}$}
        \Output{ $\mathbf{\Phi }$}

        Initialize ${\mathbf{\Phi }}{\text{ = }}{\mathbf{0}}$. \\
        Set ${\Phi _{0,j}} = 1,\forall j \in \left\{ {1,2, \ldots ,N} \right\}$ /* Assign all the SUs to the target vehicle */ \\
        \For {$i=0:{\left| {{\mathcal{V}_R}} \right|}-1$}{ 
        Check if constraint \eqref{C2} or \eqref{C1} is satisfied.\\
        \If{the related constraint is not satisfied}{ 
        \Do{ $\delta \Lambda  \leqslant 0$ }{
        Select an SU ${j^*} = \arg {\min _{j \in {\mathcal{S}_i}}}\left\{ {{\beta _j}} \right\}, \text{where\quad} {\mathcal{S}_i} = \left\{ {j\left| {{{\hat \Phi }_{i,j}} = 1,\forall j \in \left\{ {1,2, \ldots ,{N}} \right\}} \right.} \right\}$\\
        Set ${{\Phi }_{i,j^*}} = 0,{{\Phi }_{i+1,j^*}} = 1$ \\
        Calculate the remaining data amount that cannot be transmitted\\ $\delta \Lambda  = \sum\nolimits_{j = 1}^N {{\phi _{i,j}}{\beta _j}}  - \hat C_i^{{\text{V2I}}},i = 0$ \\ $\delta \Lambda  = \sum\nolimits_{j = 1}^N {{\phi _{i,j}}{\beta _j}}  - \min \left\{ {\hat C_i^{{\text{I}}},{R_{{\text{V}}}}\left( {\Delta _1^{\text{T}} + D_i^{{\text{T}}} - {{\hat \delta} _{i}}} \right)} \right\},i \in \left\{ {1,\ldots ,\left| {{\mathcal{V}_R}} \right|} \right\}$ \\
        }
        \eIf{i == 0}{
        ${{\hat \delta }_{1}} = \min \left\{ {\max \left\{ {\Delta _1^{\text{T}},\tfrac{{\sum\nolimits_{j = 1}^N {{\phi _{0,j}}{\beta _j}} }}{{{R_{{\text{I}}}}}}} \right\},{\Delta _1^{\text{T}}} + D_{1}^{{\text{T}}}} \right\}$}
{  ${{\hat \delta }_{i + 1}} = \min \left\{ {\max \left\{ {{{ \Delta}_{i + 1}^{\text{T}}},{{\hat \delta} _{i}} + \tfrac{{\sum\nolimits_{j = 1}^N {{\phi _{i,j}}{\beta _j}} }}{{{R_{{\text{V}}}}}}} \right\},{{ \Delta}_{i + 1}^{\text{T}}} + D_{i + 1}^{{\text{T}}}} \right\}$ }

        }

        }

\end{algorithm}
\subsubsection{Markov Chain Construction}
\label{Construction}
\label{MCD}
To proceed, Markov approximation implements a well-designed Markov chain with the state space of $\mathcal{F}$ to gradually  converge to the stationary distribution shown in~\eqref{25}. For any stationary distribution in product form,  there exists at least one continuous time-reversible ergodic Markov chain~\cite[Lemma 1]{chen2013markov}. Specifically, the transition rates need to meet the following two conditions:
\begin{itemize}
    \item the resulting Markov chain is irreducible, i.e., any two states are reachable from each other;
    \item The detailed balance equation is satisfied, i.e., $\forall {\mathbf{\Phi }},{\mathbf{\Phi }}' \in \mathcal{F}$, $p_{\mathbf{\Phi }}^*{q_{{\mathbf{\Phi }},{\mathbf{\Phi }}'}} = p_{{\mathbf{\Phi }}'}^*{q_{{\mathbf{\Phi }}',{\mathbf{\Phi }}}}$,
\end{itemize}
where ${q_{{\mathbf{\Phi }},{\mathbf{\Phi }}'}}$ be the transition rate from state ${\mathbf{\Phi }}$ to ${\mathbf{\Phi }}'$.
For faster convergence and easier capturing of  ${{\mathbf{\Phi }}^*}$, the Markov chain should be more likely to jump to the state with better performance.
As such, the transition rates should depend on both  $\hat U\left( {\mathbf{\Phi }} \right)$ for the current state and $\hat U\left( {{\mathbf{\Phi }}'} \right)$ for the target state.
With above in mind, the transition rate is designed as below:
\begin{equation}
{q_{{\mathbf{\Phi }},{{\mathbf{\Phi }}^\prime }}} = \frac{{\alpha \exp \left( { - \frac{{{\hat U}\left( {{{\mathbf{\Phi }}^\prime }} \right)}}{\varpi }} \right)}}{{\max \left\{ {\exp \left( { - \frac{{{\hat U}\left( {\mathbf{\Phi }} \right)}}{\varpi }} \right),\exp \left( { - \frac{{{\hat U}\left( {{{\mathbf{\Phi }}^\prime }} \right)}}{\varpi }} \right)} \right\}}},\label{26}    
\end{equation}
where $\alpha$ is a positive constant which  determines the convergence time of Markov chain. According to \eqref{26}, if ${\hat U\left( {{{\mathbf{\Phi }}^\prime }} \right) > \hat U\left( {\mathbf{\Phi }} \right)}$, the state is updated with maximum transition rate of $\alpha$. Otherwise, the larger difference between ${\hat U\left( {{{\mathbf{\Phi }}^\prime }} \right)}$ and ${\hat U\left( {\mathbf{\Phi }} \right)}$, the smaller the ${q_{{\mathbf{\Phi }},{\mathbf{\Phi '}}}}$. Moreover, the value of ${\hat U\left( {\mathbf{\Phi }} \right)}$, ${\mathbf{\Phi }} \in \mathcal{F}$, determines difference of the stationary distribution among the states, thus affecting the convergence time.
Specifically, the convergence time of the designed Markov chain is bounded as follows\footnote{The lower bound and upper bound are obtained based on spectral analysis and path coupling method, respectively. Due to space limitation, the proof process is omitted here. A similar process can be found in~\cite[Theorem 5]{chen2013markov}.}:

\noindent for $\varpi  \geqslant 2\left( {{{\hat U}_{\max }} - {{\hat U}_{\min }}} \right){\left( {\ln \left( {{{N + \tfrac{1}{{\left| {{\mathcal{V}_{\text{R}}}} \right|}}} \mathord{\left/
 {\vphantom {{N + \tfrac{1}{{\left| {{\mathcal{V}_{\text{R}}}} \right|}}} {N - 1}}} \right.
 \kern-\nulldelimiterspace} {N - 1}}} \right)} \right)^{ - 1}}$,
\begin{equation}
  {t_{{\text{mix}}}}\left( \epsilon  \right) \geqslant \frac{1}{{2\alpha M\left| {{\mathcal{V}_{\text{R}}}} \right|}}\ln \frac{1}{\epsilon }, \label{mix_1}  
\end{equation}
\begin{equation}
    {t_{{\text{mix}}}}\left( \epsilon  \right) \leqslant \frac{{\tfrac{1}{{\alpha \left| {{\mathcal{V}_{\text{R}}}} \right|}} \cdot \exp \left( {\tfrac{1}{\varpi }\left( {2{{\hat U}_{\max }} - {{\hat U}_{\min }}} \right)} \right)\ln \tfrac{N}{\epsilon }}}{{N + \tfrac{1}{{\left| {{\mathcal{V}_{\text{R}}}} \right|}} - \left( {N - 1} \right)\exp \left( {\tfrac{2}{\varpi }\left( {{{\hat U}_{\max }} - {{\hat U}_{\min }}} \right)} \right)}}, \label{mix_2} 
\end{equation}
where $\epsilon$ is the parameter to judge convergence, and ${{{\hat U}_{\max }}}$ and ${{{\hat U}_{\min }}}$ represent the maximum and minimum values of  $\hat U\left( {\mathbf{\Phi }} \right)$. According to \eqref{mix_1} and \eqref{mix_2}, we can observe that the larger the value of $\alpha$, the smaller the upper bound on the convergence time of the Markov chain. The value of $\alpha$ in our work is related to numbers of vehicles and SUs, which is specified in Section~\ref{Design}. Moreover, the differences in the value of $\hat U\left( {\bf{\Phi }} \right)$ corresponding to different ${\bf{\Phi }} \in {\cal F}$ and value of $\varpi$ also affect the convergence time.
\subsubsection{M-MTSA Design}
\label{Design}
\label{Solution}
M-MTSA is designed as shown in Fig.~\ref{algorithm2}.  M-MTSA is required to perform two functions. The one is to implement the designed Markov chain in a distributed manner. The other one is to  track the best solution during the Markov chain hopping process. For clarity, we use ${\mathbf{\Phi }}$, ${{\mathbf{\Phi '}}}$, and ${{\mathbf{\tilde \Phi }}}$ to represent the current state, the next state, and the current best strategy, respectively. It should be clarified that due to the stochastic nature of the mixed time of Markov chains, M-MTSA cannot ensure that the optimal result is obtained within $I$ iterations. In this sense, when the algorithm ends, we treat ${{{\mathbf{\tilde \Phi }}}}$ as ${{\mathbf{\Phi }}^*}$ approximately. {\color{Blues}With the aim to find a favorable solution within limited time,} we transform the continuous-time channel-hopping Markov chain to a 
discrete-time Markov chain via uniformization~\cite{ibe2013markov}. Specifically, all the threads randomly reselect another relay for their individual SUs with the probability of ${\tfrac{1}{{\left| {{\mathcal{V}_{\text{R}}}} \right|}}}$ in parallel. Then, one of the threads acquires the lock of ${\mathbf{\Phi }}$, and calculates ${\hat U\left( {\mathbf{\Phi }} \right)}$ and ${\hat U\left( {{{\mathbf{\Phi }}^\prime }} \right)}$. The state jumps from ${\mathbf{\Phi }}$ to ${{{\mathbf{\Phi }}^\prime }}$ with the probability of ${{\bar q}_{{\mathbf{\Phi }},{\mathbf{\Phi '}}}} = {{\exp \left( { - \frac{{\hat U\left( {{{\mathbf{\Phi }}^\prime }} \right)}}{\varpi }} \right)} \mathord{\left/
 {\vphantom {{\exp \left( { - \frac{{\hat U\left( {{{\mathbf{\Phi }}^\prime }} \right)}}{\varpi }} \right)} {\max \left\{ {\exp \left( { - \frac{{\hat U\left( {\mathbf{\Phi }} \right)}}{\varpi }} \right),\exp \left( { - \frac{{U\left( {{{\mathbf{\Phi }}^\prime }} \right)}}{\varpi }} \right)} \right\}}}} \right.
 \kern-\nulldelimiterspace} {\max \left\{ {\exp \left( { - \frac{{\hat U\left( {\mathbf{\Phi }} \right)}}{\varpi }} \right),\exp \left( { - \frac{{U\left( {{{\mathbf{\Phi }}^\prime }} \right)}}{\varpi }} \right)} \right\}}}$. With the assumption that each thread has an equal probability of obtaining the lock of ${\mathbf{\Phi }}$, the transition probability from ${\mathbf{\Phi }}$ to ${{{\mathbf{\Phi }}^\prime }}$ can be specified as ${q_{{\mathbf{\Phi }},{\mathbf{\Phi '}}}} = \tfrac{1}{{N\left| {{\mathcal{V}_{\text{R}}}} \right|}}{{\bar q}_{{\mathbf{\Phi }},{\mathbf{\Phi '}}}}$, which is  consistent with the form of \eqref{26}, i.e., $\alpha  = \frac{1}{{N\left| {{\mathcal{V}_{\text{R}}}} \right|}}$. 
 Meanwhile, if ${\hat U}\left( {{\mathbf{\Phi }}'} \right) > {\hat U}\left( {{{\mathbf{\tilde \Phi }}}} \right)$, the thread updates global ${{\mathbf{\tilde \Phi }}}$ to ${{\mathbf{\Phi }}'}$. 
 Assume that the optimal solution can be found after $I$ iterations. The complexity of M-MTSA is $\mathcal{O}\left( {IN} \right)$. Compared to the centralized search algorithm with the complexity of $\mathcal{O}\left( {{{\left( {\left| {{\mathcal{V}_{\text{R}}}} \right| + 1} \right)}^N}} \right)$, the complexity is greatly reduced.

Moreover, we devise an SU assignment algorithm as the baseline following the idea of the elastic-segment-based V2V/V2I cooperative strategy~\cite{liu2022elastic}, where the relays are selected in the order of encounter with the target vehicle until the requested data transmission is completed. Moreover, we treat the strategy generated by this algorithm as the initial feasible state of M-MTSA for easier capture of the optimal solution.  
 Next, we present the details about the SU assignment algorithm, which is outlined in Algorithm~1. Considering the preference for V2I links, at the beginning, all the SUs are assigned to the target vehicle $v_0$. Then, Algorithm 1 pre-checks whether the transmission of the assigned SUs can be completed, (i.e., constraint~\eqref{C1}). If not, the excessive SUs are moved to be transmitted by the next relay vehicle to encounter.  To mitigate the idleness of V2I link caused by the  non-divisibility of SU, the SUs with small data volume are moved in priority. The detail of the process is outlined in Lines 6--12. Then, Algorithm 1 calculates the transmission start moment after the encounter with the next relay vehicle, which is shown in Lines 13--17. After that, Algorithm 1 checks if constraint \eqref{C2} for the relay vehicle is satisfied. Then, Algorithm 1 repeats the above process.

\section{Simulation}
\subsection{Simulation Setup}
In the simulation, we focus on a segment of a road with two RSUs. The parameters related to the communication scenarios are summarized in Table~\ref{parameter}. {\color{Blues}In our system, the small-scale fading occurring in each transmission slot, with a duration of 1 ms, is generated by utilizing realizations of the square of the (random) small-scale channel coefficients according to~\cite[Eq.(1)]{8638956}.} The initial positions of both the target vehicle and the relay vehicles are randomly generated with a uniform distribution within $\left( { - {r_{\text{I}}},{r_{\text{I}}}} \right)$. The vehicle trajectories are generated with  SUMO, where the average routing speed and the traffic density are set 13.89 m/s and 10 vehicle/km per lane.  Moreover, vehicles can be distinguished according to the setting of attribute parameters, such as acceleration, deceleration, sigma, and maximum speed.  Three representative trajectories are shown in Fig.~\ref{SUMO}, where we can see that although the speed varies noticeably on the small time scale, the distance driven cumulatively from a  large time scale is close to the distance driven with its average speed. This  validates the rationality of analysis based on historical average speed in the proposed PreCMTS in an intuitive way. 
 % Without loss of generality, the data size and the contribution to semantic accuracy of all the SUs are randomly generated, which are listed in Table~II.  The SUs contained in the selected SR are marked in Table~II.
\begin{table*}
\renewcommand\arraystretch{1}
\scriptsize
\centering
\caption{Main Simulation Parameters.}
\label{parameter}
\begin{tabular}{|m{2.3cm} m{2.5cm}|m{2.3cm} m{2cm}|m{2.3cm} m{2cm}|}
\hline
\bf Parameters & \bf Settings & \bf Parameters & \bf Settings & \bf Parameters & \bf Settings\\
\hline
\hline
RSU coverage radius &  ${r_{\text{I}}} = 500~\text{m}$~\cite{liu2022elastic} & Vehicle coverage radius & ${r_{\text{V}}} = 300~\text{m}$~\cite{liu2022elastic} & Distance between two RSUs & $H = 1500~\text{m}$~\cite{liu2022elastic}\\ \hline
V2I channel model & ${b_{{\text{I}}}} =  1$ \newline ${a_{{\text{I}}}} = 2.2$~\cite{wu2021v2v} & V2V channel model & ${b_{{\text{V2V}}}} =  1$ \newline  ${a_{{\text{V2V}}}} = 2$~\cite{wu2021v2v}&  {\color{Blues}Average small-scale channl gain} & {\color{Blues}$\bar g = 1 \text{dB}$} \\ \hline
{\color{Blues}Fading severity} & {\color{Blues}$m=6$~\cite{8638956}} &  {\color{Blues}Shadowing shape} & {\color{Blues}$m_s = 6$~\cite{8638956}} & Link bandwidth & $B = 1~\text{MHz}$\\ \hline 
 Noise & ${\sigma ^2} =  - 110$${\text{ dBm/Hz}}$ &  SNR threshold for RSU & ${\Gamma _{\text{I}}}$ = 15.27~dB & SNR threshold for vehicle & ${\Gamma _{\text{V}}}$ = 11.44~dB\\ \hline 
 
Joint antenna gain & ${G_{\text{I}}} = {G_{\text{V}}} = 1$~\cite{chen2018cvcg} & {Reliability for direct transmission} & ${\theta _{\text{T}}} = 1.5$ &  {Reliability  for relay transmission} & ${\theta _{\text{R}}} = 0.5$ \\ \hline
\end{tabular}
\vspace{-0.2cm}
\end{table*}

\begin{table*}
    \renewcommand\arraystretch{0.8}
\scriptsize
\centering
\caption{Set of SUs.}
\label{SBset}
    \begin{tabular}{|p{0.4cm}<{\centering}|p{0.4cm}<{\centering}|p{0.4cm}<{\centering}|p{0.4cm}<{\centering}|p{0.4cm}<{\centering}|p{0.4cm}<{\centering}|p{0.4cm}<{\centering}|p{0.4cm}<{\centering}|p{0.4cm}<{\centering}|p{0.4cm}<{\centering}|p{0.4cm}<{\centering}|p{0.4cm}<{\centering}|p{0.4cm}<{\centering}|p{0.4cm}<{\centering}|p{0.4cm}<{\centering}|p{0.4cm}<{\centering}|}
    \cline{1-16}

     & \cellcolor[HTML]{C0C0C0}SU & \cellcolor[HTML]{C0C0C0}a & \cellcolor[HTML]{C0C0C0}b & \cellcolor[HTML]{C0C0C0}c & \cellcolor[HTML]{C0C0C0}d & \cellcolor[HTML]{C0C0C0}e& \cellcolor[HTML]{C0C0C0}f & \cellcolor[HTML]{C0C0C0}g & \cellcolor[HTML]{C0C0C0}h & \cellcolor[HTML]{C0C0C0}i & \cellcolor[HTML]{C0C0C0}j & \cellcolor[HTML]{C0C0C0}k & \cellcolor[HTML]{C0C0C0}l & \cellcolor[HTML]{C0C0C0}m & \cellcolor[HTML]{C0C0C0}n  \\ \cline{2-16}
      SR1   & $\beta$ & 11 & 14 & 15 & 24 & 3 & 20 & 23 & 4 & 1 & 22 & 7 & 9 & 4 & 8 \\ 
     & $\alpha$ & 0.86 & 1.08 & 1.17 & 1.87 & 0.23  & 1.56 & 1.79 & 0.31 & 0.08 & 1.71 & 0.54 & 0.70 & 0.31 & 0.62 \\ \hline
      & \cellcolor[HTML]{C0C0C0}SU & \cellcolor[HTML]{C0C0C0}a & \cellcolor[HTML]{C0C0C0}b & \cellcolor[HTML]{C0C0C0}c & \cellcolor[HTML]{C0C0C0}d &  \cellcolor[HTML]{C0C0C0}e &\cellcolor[HTML]{C0C0C0}f & \cellcolor[HTML]{C0C0C0}g & \cellcolor[HTML]{C0C0C0}h & \cellcolor[HTML]{C0C0C0}i & \cellcolor[HTML]{C0C0C0}j & \cellcolor[HTML]{C0C0C0}k & \cellcolor[HTML]{C0C0C0}l & \cellcolor[HTML]{C0C0C0}m & \cellcolor[HTML]{C0C0C0}n \\ \cline{2-16}
    \raisebox{-1.5\normalbaselineskip}[0pt][0pt]{{{SR2}}}& $\beta$ & 11 & 14 & 15 & 24 & 3 & 20 & 23 & 4 & 1 & 22 & 7 & 9 & 4 & 8  \\
      & $\alpha$ & 0.86 & 1.08 & 1.17 & 1.87 & 0.23  & 1.56 & 1.79 & 0.31 & 0.08 & 1.71 & 0.54 & 0.70 & 0.31 & 0.62 \\ \cline{2-16} 
       & \cellcolor[HTML]{C0C0C0}SU & \cellcolor[HTML]{C0C0C0}o & \cellcolor[HTML]{C0C0C0}p & \cellcolor[HTML]{C0C0C0}q & \cellcolor[HTML]{C0C0C0}r &  \cellcolor[HTML]{C0C0C0}s\\ \cline{2-7} 
       & $\beta$& 17 & 12 & 5 & 1 & 9\\ 
      &$\alpha$ & 0.06 & 0.29 & 0.62 & 0.10 & 0.59 \\\cline{1-7}

    \end{tabular}
    \label{tab:my_label}
    \vspace{-0.7cm}
\end{table*}

\begin{figure*}[t]
\begin{minipage}[t]{0.48\linewidth}
    %\begin{figure}[t]
 \centering
 \includegraphics[scale = 0.45]{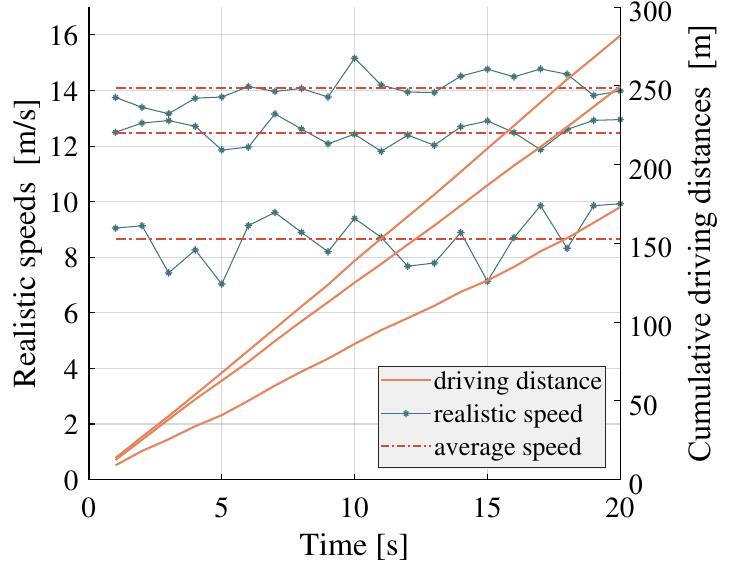}\\
 \caption{Realistic trajectory and speed generated with SUMO.
 }
 \label{SUMO}
%\end{figure}
\end{minipage}
\quad
\begin{minipage}[t]{0.48\linewidth}
    %\begin{figure}[t]
 \centering
 \includegraphics[scale = 0.45]{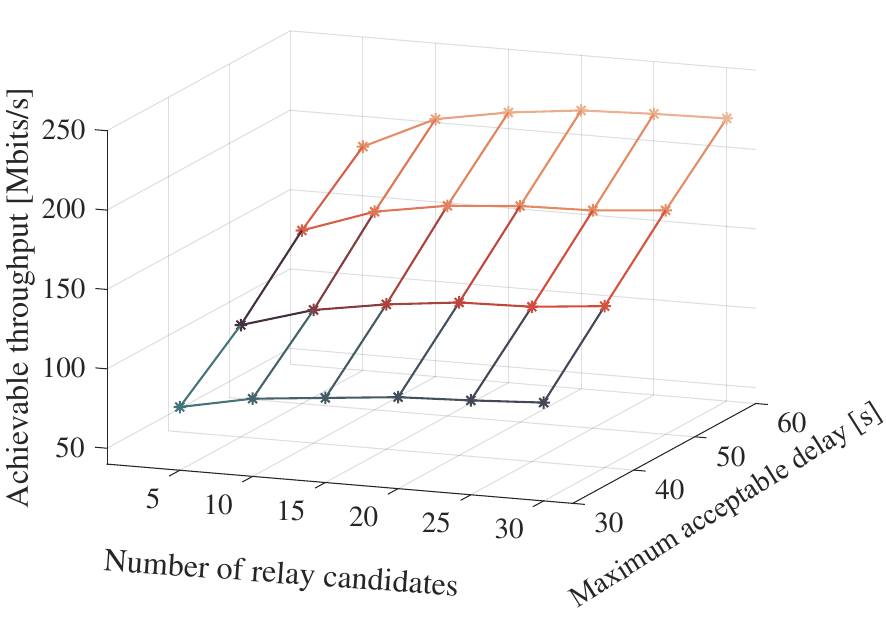}\\
 \caption{{\color{Blues}Achievable throughput with different $T_{\max}$
and $\left| {{{\mathcal{V}}_{\rm{R}}}} \right|$.}
 }
 \label{throughput}
%\end{figure}
\end{minipage}
\vspace{0cm}
\end{figure*}
\begin{table*}[t]
\centering
\renewcommand{\arraystretch}{1}
\scriptsize
\arrayrulecolor{black}
\caption{The four results of baseline and PreCMTS with different $T_{\max}$.}
\label{strategy}
\begin{threeparttable}

\begin{tabular}{|m{0.3cm}<{\centering}m{0.9cm}<{\centering}|m{0.3cm}<{\centering} m{0.9cm}<{\centering}|m{0.3cm}<{\centering} m{0.3cm}<{\centering}|m{0.3cm}<{\centering} m{0.35cm}<{\centering}|m{0.3cm}<{\centering} c|m{0.3cm}<{\centering} m{0.3cm}<{\centering}|m{0.3cm}<{\centering} m{0.35cm}<{\centering}|m{0.3cm}<{\centering} m{0.3cm}<{\centering}|m{0.3cm}<{\centering} m{0.3cm}<{\centering}|m{0.3cm}<{\centering} m{0.38cm}<{\centering}|} 
\hline
\multicolumn{2}{|c|}{Baseline}                                                                                                                                                                                                                                    & \multicolumn{4}{c|}{PreCMTS ($T_{\max}=40$ s)}                                                                                                                 & \multicolumn{6}{c|}{PreCMTS ($T_{\max}=50$ s)}                                                                                                                         & \multicolumn{8}{c|}{PreCMTS ($T_{\max}=60$ s)}                                                                                                                                                                       \\ 
\hline
Veh.                                                                                                                                                                           &SU                                                                              &Veh.                                                        &SU                                                                             &Veh.     &SU &Veh.                                                        &SU                                                                      &Veh.  &SU  &Veh.     &SU &Veh.                                                        &SU                                                                      &Veh.  &SU &Veh.     &SU &Veh.                 &SU                   \\ 
\hline
{\cellcolor[rgb]{0.886,0.886,0.886}}                                                                                                                                           & \multirow{2}{*}{\begin{tabular}[c]{@{}c@{}}a, b, c, d, \\f, 
 g, j, l\end{tabular}} & {\cellcolor[rgb]{0.886,0.886,0.886}}                        & \multirow{2}{*}{\begin{tabular}[c]{@{}c@{}}b, c, d, g,\\ j, l, n\end{tabular}} & $v_9$    & f   & {\cellcolor[rgb]{0.886,0.886,0.886}}                        & \multirow{3}{*}{\begin{tabular}[c]{@{}c@{}}b, c, \\ d, j, \\l, n\end{tabular}} & $v_1$ & h, k & $v_8$    & g   & {\cellcolor[rgb]{0.886,0.886,0.886}}                        & \multirow{3}{*}{\begin{tabular}[c]{@{}c@{}}b, c,\\ d, l,\\n\end{tabular}} & $v_1$ & h   & $v_9$    & f   & $v_{15}$             & a, j                  \\ 
\hhline{|>{\arrayrulecolor[rgb]{0.886,0.886,0.886}}-~-~>{\arrayrulecolor{black}}-->{\arrayrulecolor[rgb]{0.886,0.886,0.886}}-~>{\arrayrulecolor{black}}---->{\arrayrulecolor[rgb]{0.886,0.886,0.886}}-~>{\arrayrulecolor{black}}------|}
\multirow{-2}{*}{{\cellcolor[rgb]{0.886,0.886,0.886}}$v_0$\tnote{*}} &                                                                                  & \multirow{-2}{*}{{\cellcolor[rgb]{0.886,0.886,0.886}}$v_0$} &                                                                                 & $v_{10}$ & a   & {\cellcolor[rgb]{0.886,0.886,0.886}}                        &                                                                          & $v_3$ & e, i & $v_{12}$ & m   & {\cellcolor[rgb]{0.886,0.886,0.886}}                        &                                                                          & $v_6$ & e   & $v_{12}$ & m   & $v_{19}$             & k                     \\ 
\hhline{|------>{\arrayrulecolor[rgb]{0.886,0.886,0.886}}-~>{\arrayrulecolor{black}}---->{\arrayrulecolor[rgb]{0.886,0.886,0.886}}-~>{\arrayrulecolor{black}}------}
$v_1$                                                                                                                                                                          & e, h, i, k, m, n                                                                 & $v_1$                                                       & e, h, i, k                                                                      & $v_{12}$ & m   & \multirow{-3}{*}{{\cellcolor[rgb]{0.886,0.886,0.886}}$v_0$} &                                                                          & $v_6$ & f    & $v_{14}$ & a   & \multirow{-3}{*}{{\cellcolor[rgb]{0.886,0.886,0.886}}$v_0$} &                                                                          & $v_7$ & i   & $v_{13}$ & g   & \multicolumn{1}{c}{} & \multicolumn{1}{c}{}  \\
\hhline{|------------------~~}
\end{tabular}

\begin{tablenotes}
\item[*] {\scriptsize The shaded cells indicate the target vehicle, the other vehicle indexes indicate the selected relay vehicles in each strategy. }
\end{tablenotes}

\end{threeparttable}
\vspace{-0.5cm}
\end{table*}
\begin{figure*}[t]
  \centering
  \subfigure[]{
  \centering
  \includegraphics[scale = 0.45]{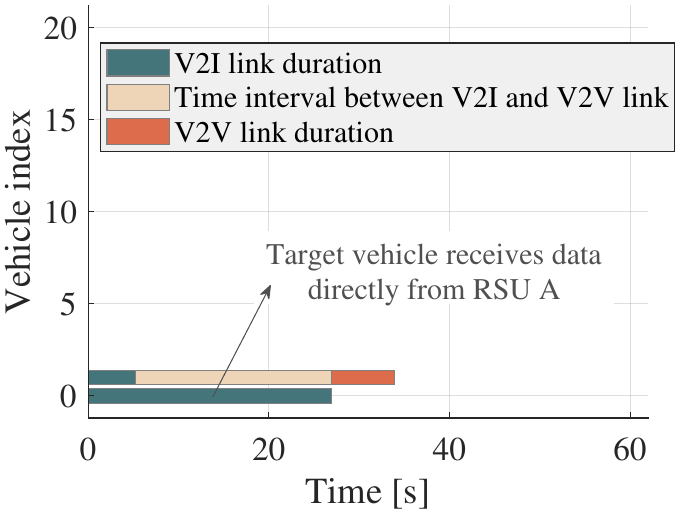}
  }
  \qquad \qquad
  \subfigure[]{
  %\begin{minipage}[t]{0.45\linewidth}
  \centering
  \includegraphics[scale = 0.45]{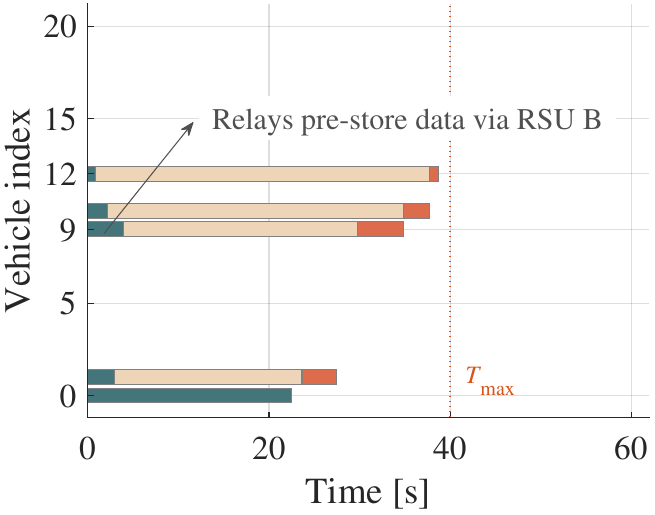}
  %\end{minipage}
  } 

\subfigure[]{
  %\begin{minipage}[t]{0.45\linewidth}
  \centering
  \includegraphics[scale = 0.45]{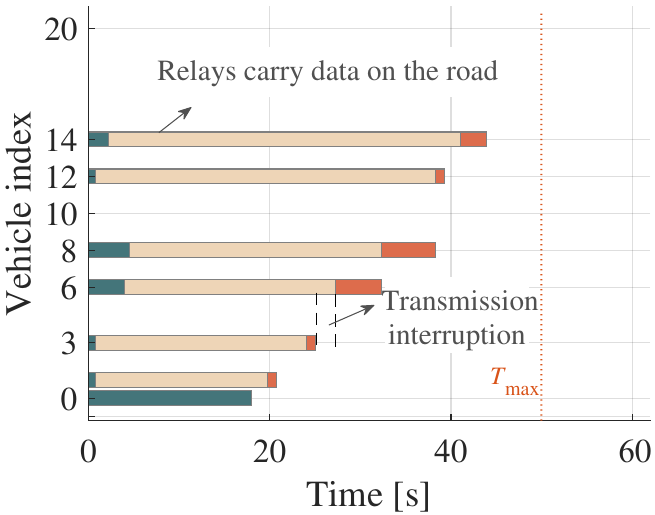}
  %\end{minipage}
  }
  \qquad \qquad
  \subfigure[]{
  %\begin{minipage}[t]{0.45\linewidth}
  \centering
  \includegraphics[scale = 0.45]{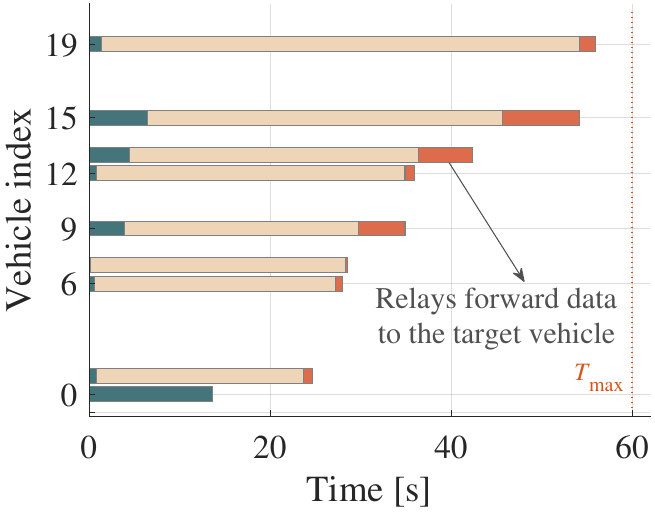}
  %\end{minipage}
  }
  \caption{ Comparison of the four strategies. (a) Baseline; (b) PreCMTS with $T_{\max} = 40$ s;  (c) PreCMTS with $T_{\max} = 50$ s; (d) PreCMTS with $T_{\max} = 60$ s.
  }
  \label{Comparisonduration} 
  \vspace{-0.2cm}
\end{figure*}

\subsection{Performance Evaluation}
The achievable throughput within $T_{\max}$ not only determines  the maximum data volume of the SR that can be supported for the current scenario, but also the optimizable space of PreCMTS given a selected SR.
In this sense, before the evaluation of the proposed PreCMTS, we first show the average achievable throughput of 50 simulations with randomly generated initial positions under different maximum acceptable delay $T_{\max}$ and different numbers of relay candidates $\left| {{\mathcal{V}_{\text{R}}}} \right|$ in Fig.~\ref{throughput}.  From Fig.~\ref{throughput}, we can observe that as the number of relay candidates increases and the delay requirement is relaxed, there is an increase in the achievable throughput. Specifically, with the increase of the maximum acceptable delay, the rise in the achievable throughput, as the number of vehicles in RSU B increases, is more  significant. This is because the vehicles are randomly distributed within RSU B. When the value of $T_{\max}$ is small, a high percentage of vehicles encounter the target vehicle exceed the maximum acceptable delay, which fails to contribute to the throughput. Moreover, since only one V2V link can exist at any moment, when the number of vehicle candidates reaches a certain threshold, the increase in throughput becomes flat. Furthermore,  the threshold value of the number of relay candidates decreases with the increase of the maximum acceptable delay.

Next,  we take take an example with 20 relay candidates to evaluate the performance of the  PreCMTS with ${\kappa _1} = 0.5,{\kappa _2} = 0.1$. To demonstrate its superiority, we evaluate the PreCMTS under different maximum acceptable delay, i.e.,  ${T_{\max }} = 40{\text{ s}}$, ${T_{\max }} = 50{\text{ s}}$, and ${T_{\max }} = 60{\text{ s}}$ with the baseline devised according to~\cite{liu2022elastic}. Based on \eqref{Qmax}, we have ${Q_{\max }} = 186.6$~Mbits within 40 s. Accordingly, with out loss of generality, we generate a SR indexed by SR 1 as Table~\ref{SBset} with the total volume of 165 Mbits and semantic accuracy of 12.84 to perform the simulation. The four strategies derived in the above four cases are presented in Table.~\ref{strategy}, where the SUs within the SR 1 that each vehicle should transmit are determined.  A more intuitive presentation is in Fig.~\ref{Comparisonduration}. Moreover, the practical trajectory information of all the relay candidates are shown in Table.~\ref{Distance} and Fig.~\ref{V2V}. At last, the performance of PreCMTS in terms of energy efficient and semantic  reliability are shown in Fig.~\ref{SSSS}, respectively.

\begin{figure*}
  \centering
 \subfigure[]{
  %\begin{minipage}[t]{0.45\linewidth}
  \centering
  \includegraphics[scale = 0.37]{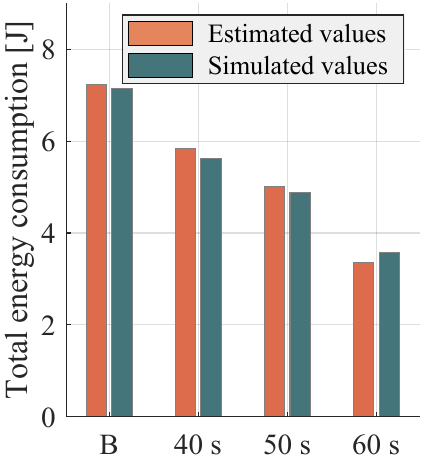}
  %\end{minipage}
  }  
   \centering
 \subfigure[]{
  %\begin{minipage}[t]{0.45\linewidth}
  \centering
  \includegraphics[scale = 0.37]{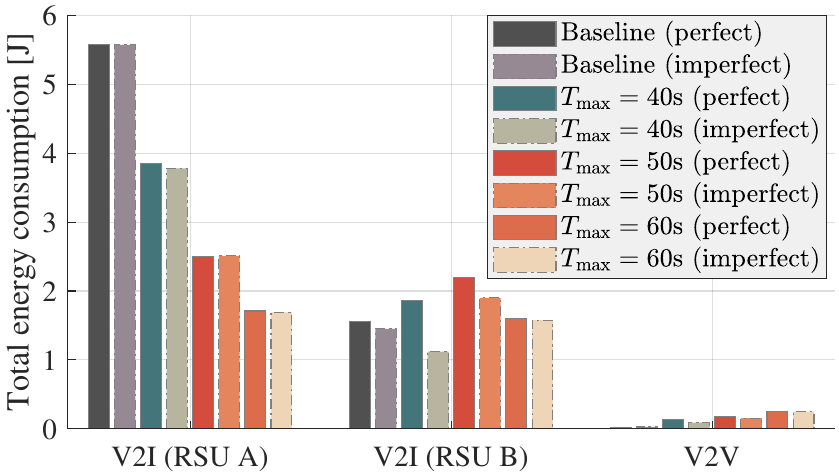}
  %\end{minipage}
  }  
  \centering
 \subfigure[]{
  %\begin{minipage}[t]{0.45\linewidth}
  \centering
  \includegraphics[scale = 0.37]{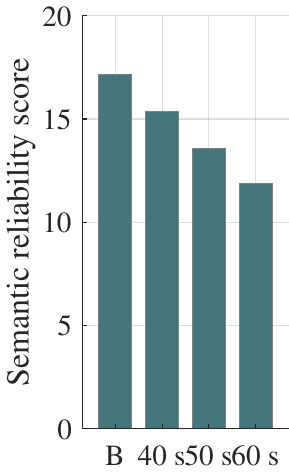}
  %\end{minipage}
  }
\subfigure[]{
  %\begin{minipage}[t]{0.45\linewidth}
  \centering
  \includegraphics[scale = 0.37]{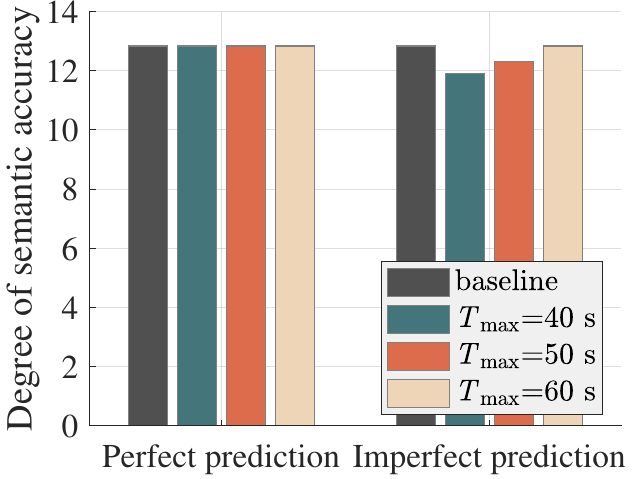}
  %\end{minipage}
  }
  \caption{ {\color{Blues}Performance analysis. (a) Comparison of theoretical and simulated values of cumulative energy consumption. (b) Comparison of energy consumption of the strategies. }(c) Comparison of semantic reliability scores; (d) Comparison of degree of semantic accuracy.
  }
  \label{SSSS}
\end{figure*}

{\color{Blues}Overall, as shown in Fig.~\ref{SSSS}(a), the simulated cumulative energy consumption basically coincides the theoretical estimated value, which supports the rationality of performance evaluation in terms of energy consumption. Moreover, it is evident that the overall energy consumption experiences a significant reduction when the delay requirements are relaxed following the optimization of the proposed PreCMTS. Specifically,} by comparing Fig.~\ref{Comparisonduration} and Table~\ref{Distance}, we can observe that in the baseline, the target vehicle consistently maintains the V2I link within the coverage area of RSU A. Owing to  ${R_{\text{I}}} > R_{\text{V}}$, the baseline achieves the lowest transmission delay. However, due to the low channel gain at the edge of coverage, such a V2I link-first mechanism would cause significant energy waste. In our PreCMTS, by optimizing the relay selection and SU assignment within $T_{\max}$, it selectively assigns partial SUs to the store-carry-forward links. In this way, the energy efficiency  in RSU A can remarkably increase, which can be verified by calculating the ratio of the V2I link duration of $v_0$ (or the data volume assigned to $v_0$)  and the  energy consumption of V2I link in RSU~A  in the four strategies according to Figs.~\ref{Comparisonduration} and~\ref{SSSS}(b). Moreover, form Fig.~\ref{Comparisonduration}(b)-(d), with the relaxation of time delay requirements,  the relays closer to RSU B, such as  $v_{14}$, and $v_{15}$ are more likely to be selected under different delay requirements to enhance energy efficiency. This also allows the fact shown in Fig.~\ref{SSSS}(b) that the total energy consumption of the V2I links in RSU B does not increase monotonically with the  data volume transmitted. Meanwhile, the vehicles far from the RSU B such as $v_2$ and $v_4$ are missed in 
the all the strategies, even if  this leads to an avoidable transmission interruption implied by Fig.~\ref{V2V}. This is the key reason why the PreCMTS can achieve higher energy efficiency compared with the existing schemes.  Moreover, comparing Figs.~\ref{Comparisonduration} and~\ref{V2V}, most V2V links are established after a period of time when the relay encounters the target vehicle. This means that the transmission distance of V2V links is generally shorter and thus more energy can be saved. Therefore, the energy consumption of V2V links is significantly smaller than that of the V2I links. 
Moreover, as shown in Fig.~\ref{SSSS}(c), as more SUs are assigned to the relays, the semantic reliability score becomes smaller, which is consistent with the definition in \eqref{DDD}. However, since we minimize energy consumption while optimizing semantic reliability, SUs with less semantic importance are preferentially assigned to other relays, which can be seen by comparing Tables~\ref{SBset} and \ref{strategy}. Since  the SUs with large semantic importance are mostly assigned to the direct link, the transmission of them is completed with priority. Therefore, in some cases with imperfect speed prediction, there is no remarkable decrease in the degree of the semantic accuracy as shown in Fig.~\ref{SSSS}(d). Moreover, in some cases, the SR can still be fully transmission with imperfect speed prediction, but it consumes more energy. This shows the adaptability of the proposed strategy to sudden changes in the vehicular environment.

\begin{table}[]
    \centering
    \captionof{table}{\color{Blues}Initial locations and average speeds of vehicles.}
\centering%
\renewcommand\arraystretch{0.8}
\scriptsize
\begin{tabular}{|m{0.3cm}<{\centering}|m{0.5cm}<{\centering}|m{0.5cm}<{\centering}|m{0.5cm}<{\centering}|m{0.5cm}<{\centering}|m{0.5cm}<{\centering}|m{0.5cm}<{\centering}|m{0.5cm}<{\centering}|} 
\hline
\rowcolor[rgb]{0.812,0.816,0.816} Veh. & $v_0$    & $v_1$    & $v_2$    & $v_3$    & $v_4$    & $v_5$    & $v_6$     \\ 
\hline
$l_i$                                  & 200   & 382   & 484   & 403   & 438  & 340 & 336    \\ 
\hline
${\bar u}_i$                           & 10.97    & 15.44    & 10.81    & 14.14    & 11.28    & 13.31    & 13.41     \\ 
\hline
\rowcolor[rgb]{0.812,0.816,0.816} Veh. & $v_7$    & $v_8$    & $v_9$    & $v_{10}$ & $v_{11}$ & $v_{12}$ & $v_{13}$  \\ 
\hline
$l_i$                                  & 317   & 260  & 308   & 214   & 253   & 220   & 281    \\ 
\hline
${\bar u}_i$                           & 13.10    & 14.03    & 12.30    & 13.30    & 11.41    & 11.81    & 8.89      \\ 
\hline
\rowcolor[rgb]{0.812,0.816,0.816} Veh. & $v_{14}$ & $v_{15}$ & $v_{16}$ & $v_{17}$ & $v_{18}$ & $v_{19}$ & $v_20$    \\ 
\hline
$l_i$                                  & 0.12     & $-39$   & $-50$   & $-10$   & $-112$  & $-202$  & $-254$   \\ 
\hline
${\bar u}_i$                           & 13.63    & 12.38    & 11.80    & 10.72    & 12.90    & 13.45    & 13.53     \\
\hline
\end{tabular}
\label{Distance}
\end{table}

\begin{figure}
    \centering
    \includegraphics[width=0.4\textwidth]{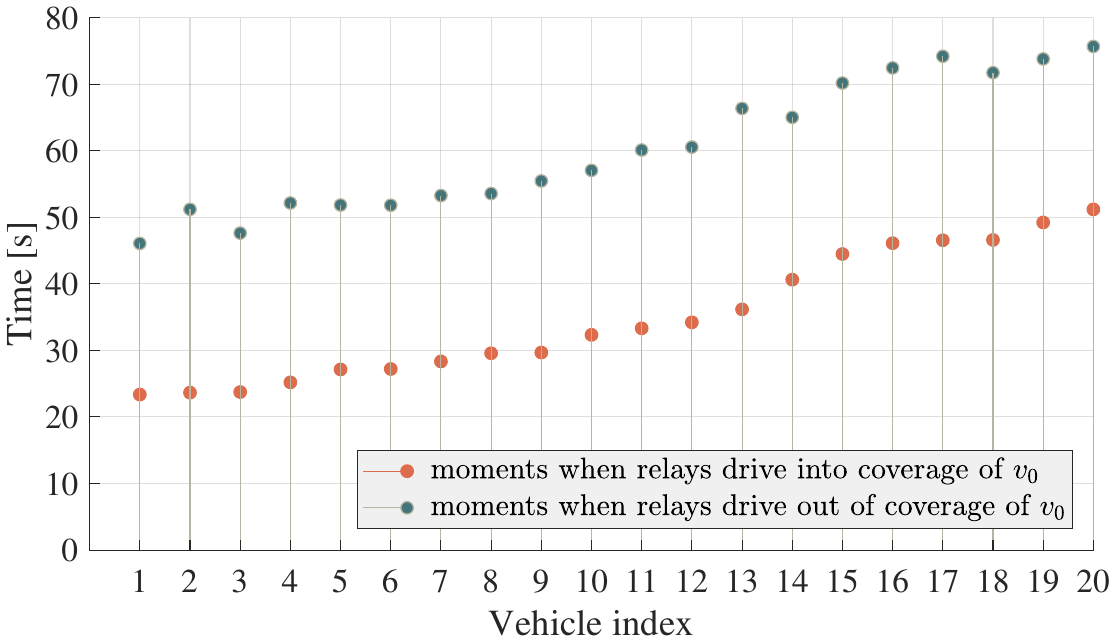}
\captionof{figure}{Encounter time between relay and target vehicle.}
\label{V2V}
\end{figure}

In addition, by adjusting the values of ${\kappa _1}$ and ${\kappa _2}$, the attention of PreCMTS to the energy consumption and semantic reliability of SR transmission can be adjusted. In Fig.~\ref{kkk1}, we compare  two PreCMTS under the settings of  ${\kappa _1} = 0.5$, ${\kappa _2} = 0.1$, and  {\color{Blues}${\kappa _1} = 0.1$, ${\kappa _2} = 0.5$}, respectively. The two specific strategy results can be found in Table.~\ref{strategy} and Table~\ref{strategy2}. 
As shown in Fig.~\ref{kkk1}(a), with the increase of the value of ${\kappa _2}$, the semantic reliability score is improved significantly, and the energy consumption is increased slightly. Moreover, as the SU with high semantic importance such SU $g$ and SU $f$ are assigned to $v_0$, the degree of semantic accuracy achieved by PreCMTS S is higher than that achieved by PreCMTS E, as shown in Fig.~\ref{kkk1}(c). In addition, it is to be noted that the above simulation only supports  that the values of ${\kappa _1}$ and ${\kappa _2}$ can effectively influence the strategy results. However, as there is no definite linear relationship between the semantic importance of an SU and its data volume, the optimal combination of values for ${\kappa _1}$ and ${\kappa _2}$ deserves further investigation.

\begin{table*}
\centering
\arrayrulecolor{black}
\renewcommand{\arraystretch}{1}
\scriptsize
\caption{The three results of PreCMTS under different SRs and $T_{\max}$.}
\label{strategy2}
\begin{threeparttable}
\begin{tabular}{|p{0.6cm}<{\centering}|c|p{0.6cm}<{\centering}|p{1.8cm}<{\centering}|p{0.6cm}<{\centering}|p{1cm}<{\centering}|p{0.6cm}<{\centering}|c|p{0.6cm}<{\centering}|p{0.6cm}<{\centering}|p{0.6cm}<{\centering}|p{0.6cm}<{\centering}|} 
\hline
\multicolumn{4}{|c|}{PreCMTS (${\kappa _1} = 0.1;{\kappa _2} = 0.5$)}                                                                                                                                                                              & \multicolumn{8}{c|}{PreCMTS ($T_{\max}=50$s, SR 2)}                                                                                                                                                                                                                                                                                                                                                                                               \\ 
\hline
Veh.                                                        & SU                                                                             & Veh.                      & SU                                                                      & Veh.                                                           & SU                                                                         & Veh.                       & SU                                                                & Veh.                          & SU                                                                & Veh.                      & SU                                                                 \\ 
\hline
{\cellcolor[rgb]{0.847,0.847,0.847}}                        & \multirow{2}{*}{\begin{tabular}[c]{@{}c@{}}b, c, d, f, \\g, j, l, n\end{tabular}} & \multirow{2}{*}{$v_{12}$} & \multirow{2}{*}{m}                                                      & {\cellcolor[rgb]{0.847,0.847,0.847}}                           & \multirow{2}{*}{\begin{tabular}[c]{@{}c@{}}b, c, d,\\f, g, j\end{tabular}} & \multirow{2}{*}{$v_3$}     & \multirow{2}{*}{\begin{tabular}[c]{@{}c@{}}m,\\n, q\end{tabular}} & \multirow{2}{*}{$v_4$}        & \multirow{2}{*}{\begin{tabular}[c]{@{}c@{}}e, h,\\r\end{tabular}} & $v_{14}$                  & o                                                                  \\ 
\hhline{|>{\arrayrulecolor[rgb]{0.847,0.847,0.847}}-~~~-~~~~~>{\arrayrulecolor{black}}--|}
\multirow{-2}{*}{{\cellcolor[rgb]{0.847,0.847,0.847}}$v_0$\tnote{*}} &                                                                                &                           &                                                                         & \multirow{-2}{*}{{\cellcolor[rgb]{0.847,0.847,0.847}}$v_0$}    &                                                                            &                            &                                                                   &                               &                                                                   & $v_{16}$                  & a                                                                  \\ 
\hline
$v_1$                                                       & e, h, i, k                                                                     & $v_{14}$                  & a                                                                   & $v_1$                                                          & p, s                                                                       & $v_9$                      & k                                                                 & $v_{10}$                      & l                                                                 & $v_{17}$                  & i                                                                  \\ 
\hline
\multicolumn{4}{|c|}{Baseline (SR2)   }                                                                                                                                                                                                                & \multicolumn{8}{c|}{PreCMTS ($T_{\max}=60$s, SR 2)}                                                                                                                                                                                                                                                                                                                                                                                               \\ 
\hline
Veh.                                                        & SU                                                                             & Veh.                      & SU                                                                      & Veh.                                                           & SU                                                                         & Veh.                       & SU                                                                & Veh.                          & SU                                                                & Veh.                      & SU                                                                 \\ 
\hline
{\cellcolor[rgb]{0.847,0.847,0.847}}                        & \multirow{2}{*}{\begin{tabular}[c]{@{}c@{}}b, c, d,\\f, g, j, o\end{tabular}}  & \multirow{2}{*}{$v_3$~}   & \multirow{2}{*}{\begin{tabular}[c]{@{}c@{}}k, l, m,\\n, q\end{tabular}} & {\cellcolor[rgb]{0.847,0.847,0.847}}                           & \multirow{2}{*}{\begin{tabular}[c]{@{}c@{}}b, c, \\d, f\end{tabular}}      & $v_3$                      & m, n                                                              & $v_{11}$                      & i                                                                 & \multirow{2}{*}{$v_{16}$} & \multirow{2}{*}{\begin{tabular}[c]{@{}c@{}}a, e,\\j\end{tabular}}  \\ 
\hhline{|>{\arrayrulecolor[rgb]{0.847,0.847,0.847}}-~~~-~>{\arrayrulecolor{black}}----~~|}
\multirow{-2}{*}{{\cellcolor[rgb]{0.847,0.847,0.847}}$v_0$} &                                                                                &                           &                                                                         & \multirow{-2}{*}{{\cellcolor[rgb]{0.847,0.847,0.847}}$v_{0}$~} &                                                                            & \multicolumn{1}{c|}{$v_4$} & \multicolumn{1}{c|}{h, r}                                         & \multicolumn{1}{c|}{$v_{13}$} & o                                                                 &                           &                                                                    \\ 
\hline
$v_1$                                                       & a, p, s                                                                        & $v_4$                     & e, h, i, r                                                              & $v_{1}$                                                        & p, s                                                                       & $v_9$                      & k, l                                                              & $v_{14}$                      & g                                                                 & $v_{20}$                  & q                                                                  \\
\hline
\end{tabular}

\begin{tablenotes}
\item[*] {\scriptsize The shaded cells indicate the target vehicle, the other vehicle indexes indicate the selected relay vehicles in each strategy.}
\end{tablenotes}

\end{threeparttable}
\end{table*}
\begin{figure*}[t]
 \centering
 \subfigure[]{
  %\begin{minipage}[t]{0.45\linewidth}
  \centering
  \includegraphics[scale = 0.4]{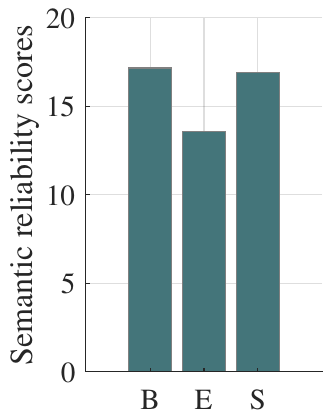}
  %\end{minipage}
  }
  \quad
\subfigure[]{
  %\begin{minipage}[t]{0.45\linewidth}
  \centering
  \includegraphics[scale = 0.4]{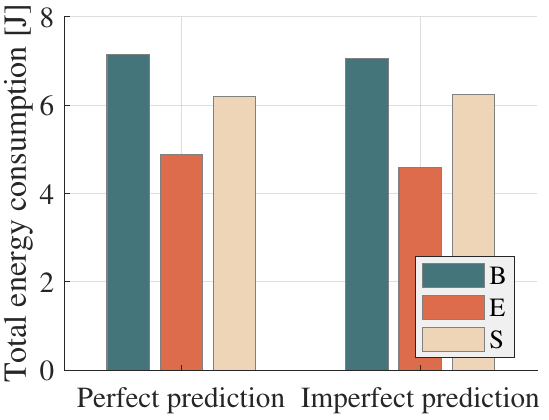}
  %\end{minipage}
  }
  \quad
  \subfigure[]{
  %\begin{minipage}[t]{0.45\linewidth}
  \centering
  \includegraphics[scale = 0.4]{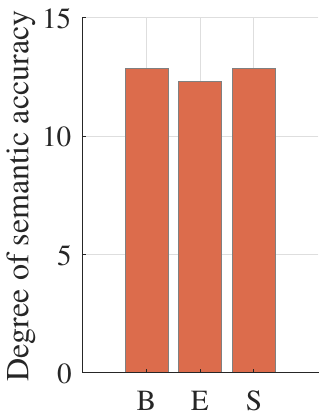}
  %\end{minipage}
  }
  \caption{\color{Blues} Comparison of PreCMTS with different ${\kappa _1}$ and ${\kappa _2}$, where ``B" represents the baseline, ``E" represents the PreCMTS with ${\kappa _1} = 0.5$ and ${\kappa _2} = 0.1$, and ``S" represents the PreCMTS with ${\kappa _1} = 0.1$ and ${\kappa _2} = 0.5$. (a) Semantic reliability scores; (b) Energy consumption; (c) Degree of semantic accuracy with imperfect speed prediction.
  }  \label{kkk1}
*
\end{figure*}

Furthermore, we evaluate the performance of PreCMTS under different SRs. {\color{Blues}According to \eqref{Qmax}, when the maximum acceptable delay extends to 50 s and 60 s, $Q_{\max} = 225.6$ Mbits and $Q_{\max} = 264.6$ Mbits, respectively.} We generate a second SR indexed by SR 2 as shown in Table.~\ref{SBset}. Considering that the SUs with small semantic contribution is filtered out in priority, the SUs with relative small value of ${\alpha _j}$ are added in SR 2.  Moreover, for a intuitive performance comparison, we define a new metric called semantic energy efficiency as the ratio of the degree of semantic accuracy and total energy consumption, i.e., ${\text{E}}{{\text{E}}_{\text{S}}} = {{\sum\nolimits_{j = 1}^N {{\alpha _j}} } \mathord{\left/
 {\vphantom {{\sum\nolimits_{j = 1}^N {{\alpha _j}} } {\left( {{P_{{\text{V2V}}}} + {P_{{\text{V2I}}}}} \right)}}} \right.
 \kern-\nulldelimiterspace} {\left( {{P_{{\text{V2V}}}} + {P_{{\text{V2I}}}}} \right)}}$. As shown in Fig.~\ref{SEE}, as the added SUs is with relatively small semantic importance and random data volume, the semantic energy efficiency of the strategies with SR 2 is clearly lower than that achieved by the corresponding strategies with SR 1.  Meanwhile, with the extension of the acceptable delay, the semantic energy efficiency increases remarkably for both SR 1 and SR 2. Specifically, the semantic energy efficiency of the PreCMTS with SR 1 has a greater enhancement than that with SR 2. This means that the PreCMTS achieves superior performance at lower system loads. Thus, a trade-off between the energy efficiency and semantic accuracy degree can be made on a case-by-case basis.

\begin{figure}
     \centering
 \includegraphics[scale = 0.45]{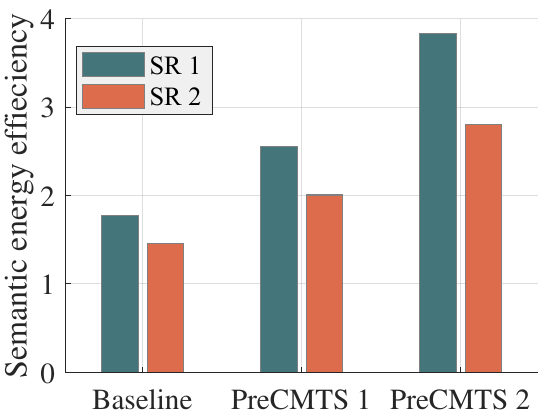}\\
 \caption{\color{Blues}Comparison of the PreCMTS with SR 1 and SR 2, where PreCMTS 1 represents the PreCMTS with $T_{\max} = 50$~s, and PreCMTS 2 represents the PreCMTS with $T_{\max} = 60$~s. }
 \label{SEE}
 \vspace{-0.6cm}
\end{figure}

\section{Conclusion}
In this paper, we have proposed a predictive cooperative multi-relay transmission strategy for bidirectional road scenarios. Specifically, we have introduced a general task-driven KG-assisted SemCom system  for complex vehicular network. To facilitate semantic-aware transmission, we have modeled the KG into a wDG. Next, for an appropriate SR, we  have
derived the closed-form expression for the achievable throughput for  within the maximum acceptable delay. Moreover,  we have formulated the relay vehicle selection and SU assignment as  a combinatorial optimization problem to optimize energy efficiency and semantic reliability. To finding a favorable solution within limited time, we have solved the problem with a low-complexity M-MTSA based on Markov approximation, where the solution is iteratively optimized. To demonstrate the feasibility of the PreCMTS, we have simulated it with realistic vehicle traces generated by SUMO. The high energy efficiency, semantic transmission reliability, and semantic energy efficiency of PreCMTS have been demonstrated with simulations.

% Parameter

% \cite{wu2021v2v}

% B = 1~MHz

% $\alpha_{\text{V2I}}$ = 2.2

% $\alpha_{\text{V2I}}$ = 2

% \cite{xu2021socially}

% Noise = -90dB

% \cite{}
% $r_\text{I}$ = 500~m

% $r_\text{V}$ = 300~m

% zhou2019reliability

\bibliographystyle{IEEEtran}
\bibliography{Semantic-aware-PRTS}

\begin{IEEEbiography}[{\includegraphics[width=1in,height=1.25in,clip,keepaspectratio]{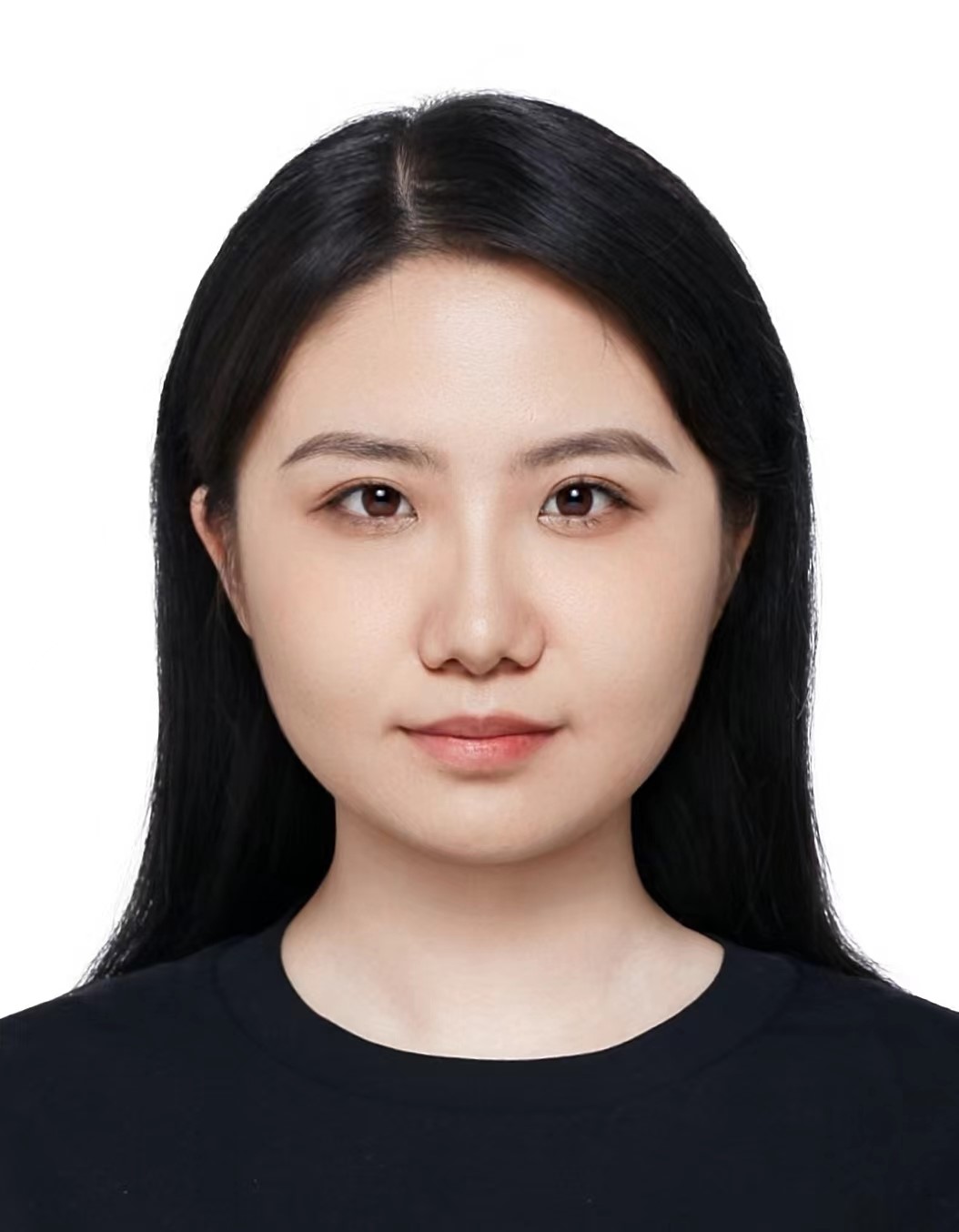}}]{Wanting Yang}
% or if you just want to reserve a space for a photo:

%\begin{IEEEbiography}{Michael Shell}
received the B.S. degree and the Ph.D. degree from the Department of Communications Engineering, Jilin University, Changchun, China, in 2018 and 2023, respectively. She was a visiting student at Singapore University of Technology and Design from 2021 to 2022, sponsored by the Chinese Scholarship Council. She
served as Technical Programme Committee member in flagship
conferences, such as WCNC, Globecom, and VTC. Her research interests include wireless communications, predictive resource allocation, semantic communication, learning, martingale and URLLC.
\end{IEEEbiography}

\begin{IEEEbiography}[{\includegraphics[width=1in,height=1.25in,clip,keepaspectratio]{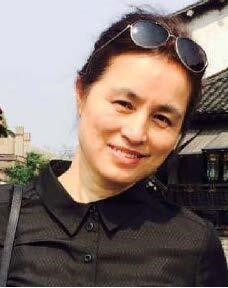}}]{Xuefen Chi}
received the B.Eng. degree in applied physics from the Beijing University of Posts and
Telecommunications, Beijing, China, in 1984, and the M.S. and Ph.D. degrees from the Changchun Institute of Optics, Fine Mechanics and Physics, Chinese Academy of Sciences, Changchun, China, in 1990 and 2003, respectively. She was a Visiting Scholar with the Department of Computer Science, Loughborough University, U.K., in 2007, and the School of Electronics and Computer Science, University of Southampton, Southampton, U.K., in 2015. She is currently a Professor with the Department of Communications Engineering, Jilin University, China. Her research interests include machine-type communications, indoor visible light communications, random access algorithms, delay-QoS guarantees, and network modeling theory and its applications.
\end{IEEEbiography}

\begin{IEEEbiography}[{\includegraphics[width=1in,height=1.25in,clip,keepaspectratio]{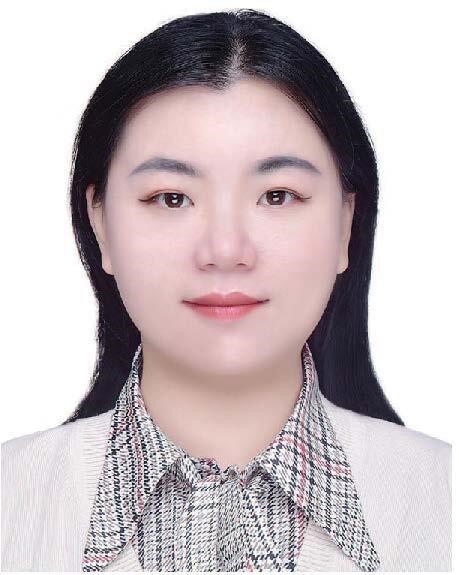}}]{Linlin Zhao}
received the B.Eng., M.S., and Ph.D. degrees from the Department of Communications Engineering, Jilin University, Changchun, China, in 2009, 2012, and 2017, respectively. From 2017 to 2019, she was a Post-Doctoral Researcher with the Department of Communications
Engineering, Jilin University. She is currently an Associate Professor with the Department of Communications Engineering, Jilin University, and a Post-Doctoral Research Fellow with the State Key Laboratory of Internet of Things for Smart City, University of Macau. Her current research interests include throughput optimal random access algorithms, resource allocation schemes, and delay and reliability analysis and optimization, especially for reliability analysis of ultra-reliable low-latency communications. She was a recipient of the Best Ph.D. Thesis Award of Jilin University in 2017, and acquired the Macau Young Scholars Program in 2019. She has served as the Registration Co-Chair for IEEE ICCC 2019.
\end{IEEEbiography}

\begin{IEEEbiography}[{\includegraphics[width=1in,height=1.25in,clip,keepaspectratio]{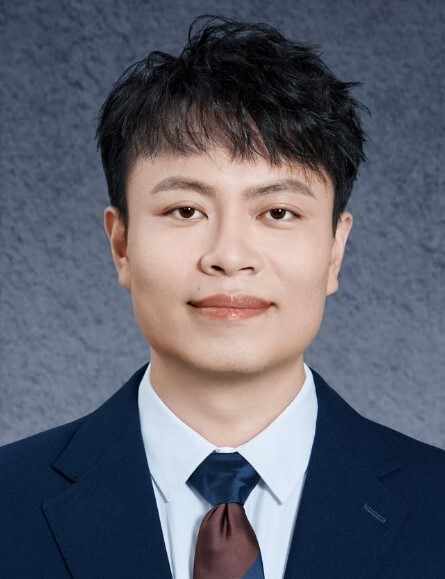}}]{Zehui Xiong}
is currently an Assistant Professor at Singapore University of Technology and Design, and also an Honorary Adjunct Senior Research Scientist with Alibaba-NTU Singapore Joint Research Institute, Singapore. He received the PhD degree in Nanyang Technological University (NTU), Singapore. He was the visiting scholar at Princeton University and University of Waterloo. His research interests include wireless communications, Internet of Things, blockchain, edge intelligence, and Metaverse. He has published more than 200 research papers in leading journals and flagship conferences and many of them are ESI Highly Cited Papers. He has won over 10 Best Paper Awards in international conferences and is listed in the World’s Top 2\% Scientists identified by Stanford University. He is now serving as the editor or guest editor for many leading journals including IEEE Journal on Selected Areas in Communications, IEEE Transactions on Vehicular Technology, IEEE Internet of Things Journal, IEEE Transactions on Cognitive Communications and Networking, and IEEE Transactions on Network Science and Engineering. He is the recipient of IEEE Early Career Researcher Award for Excellence in Scalable Computing, IEEE Technical Committee on Blockchain and Distributed Ledger Technologies Early Career Award, IEEE Internet Technical Committee Early Achievement Award, IEEE TCSVC Rising Star Award, IEEE TCI Rising Star Award, IEEE TCCLD Rising Star Award, IEEE Best Land Transportation Paper Award, IEEE CSIM Technical Committee Best Journal Paper Award, IEEE SPCC Technical Committee Best Paper Award, IEEE VTS Singapore Best Paper Award, Chinese Government Award for Outstanding Students Abroad, and NTU SCSE Best PhD Thesis Runner-Up Award. He is now serving as the Associate Director of Future Communications R\&D Programme. In 2023, he was featured on the list of Forbes Asia 30 under 30.
\end{IEEEbiography}

\begin{IEEEbiography}[{\includegraphics[width=1in,height=1.25in,clip,keepaspectratio]{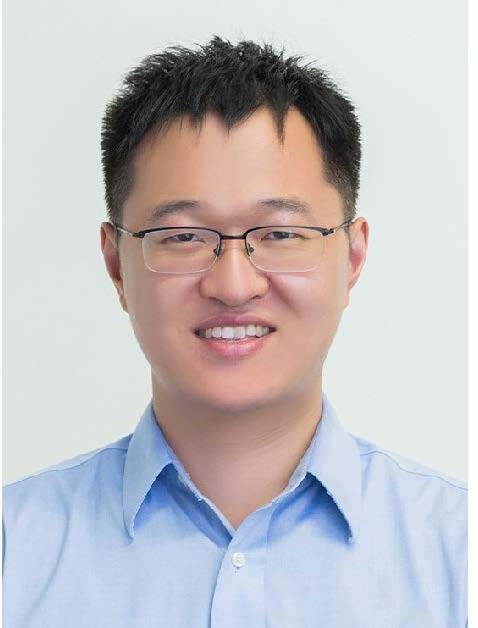 }}]{Wenchao Jiang}
received the Ph.D. degree from the Department of Computer Science and Engineering, University of Minnesota Twin Cities, in 2019. He is currently an Assistant Professor
with the Pillar of Information System Technology and Design, Singapore University of Technology and Design. His research interests include the Internet of Things, wireless and low-power embedded networks, and mobile computing.
\end{IEEEbiography}

% \cite{}
% \section*{Biographies}
% \small

% {Author} is with

\end{document}